\documentclass[12pt]{article}

\newif\iflongversion
\longversiontrue
\usepackage{fullpage}
\usepackage{times}

\usepackage{doi}

\usepackage{amsthm}

\theoremstyle{plain}
\newtheorem{theorem}{Theorem}
\newtheorem{lemma}[theorem]{Lemma}

\theoremstyle{definition}
\newtheorem{definition}{Definition}
\theoremstyle{remark}

\usepackage{amsmath}
\usepackage{paralist}

\RequirePackage{stmaryrd} %

\usepackage[english]{babel}
\usepackage{math}
\usepackage[prefixflatinterpret,bracketmodalinterpret,fixformat,sidenotecalculus,longseqcontext,bbmquant,colormquant,seqarrow,seqoptional]{logic}
\usepackage[prefixflatinterpret,bracketmodalinterpret,fixformat,differentialdL]{dL}
\renewcommand{\ltrue}{\top}
\newcommand{\pusall}{\ddiamond{c}{\ltrue}}
\newcommand{\qusall}{\ddiamond{d}{\ltrue}}

\renewcommand{\mexists}[2][]
{\textcolor{blue}{\exists}#2\,}

\def\qedhere{\qed}
\usepackage[keepproof]{endproof2}

\def\qedhere{}

\usepackage{tikz}
\usetikzlibrary{shapes}

\usepackage{xcolor}
\definecolor{vgreen}{rgb}{.1,.5,0}
\definecolor{vred}{rgb}{.7,0,0}
\definecolor{vblue}{rgb}{.1,.15,.62}

\tikzstyle{transition}=[-stealth]
\tikzstyle{similar state}=[double]
\tikzstyle{transition exists}=[dashed]
\tikzstyle{equivalence}=[<->,double]

\tikzstyle{strategy preimage}=[draw=vblue,fill=vblue!20,shape=ellipse]
\tikzstyle{winning condition}=[draw=vred,fill=vred!50,text=white,shape=circle]
\tikzstyle{winning annotation}=[ultra thick,black,->]

\renewcommand{\allvars}{\mathbf{V}}%

\renewcommand*{\applyusubst}[2]{#1{#2}}%
\newcommand*{\usubstgroup}[1]{(#1)}

\renewcommand{\with}{\mathrel{:}}
\newcommand*{\ignore}[1]{}

\newcommand*{\linterpretationsconst}[2]{\mathcal{I}}
\renewcommand{\linterpretationsname}{\mathcal{S}}
\newcommand{\I}{\vdLint[const=I,state=\omega]}
\newcommand{\It}{\vdLint[const=I,state=\nu]}%
\newcommand{\Iz}{\vdLint[const=I,state=\mu]}
\newcommand{\If}{\DALint[const=I,flow=\varphi]}
\newcommand*{\Iff}[1][\zeta]{\vdLint[const=I,state=\varphi(#1)]}%

\newcommand{\Ia}{\iadjointSubst{\sigma}{\I}}%
\newcommand{\Iminner}{\imodif[const]{\Ie}{\,\usarg}{d}}%
\newcommand{\Iminnerorg}{\imodif[const]{\I}{\,\usarg}{d}}%
\newcommand{\Iiminner}{\imodif[const]{\I}{\,\usarg}{d}}%

\newcommand{\Ie}{\vdLint[const=I,state=\nu]}
\newcommand{\Iae}{\iconcat[state=\nu]{\Ia}}

\newcommand{\IIalt}{\vdLint[const=I,state=\tilde{\omega}]}%
\let\Ialt\IJalt

\newcommand{\Iaz}{\iconcat[state=\mu]{\Ia}}%

\newcommand{\bebecomes}{\mathrel{::=}}
\newcommand{\alternative}{~|~}
\providecommand{\dfn}[2][]{\emph{#2}}

\newcommand*{\iwinreg}[3][]{\iaccess[#1]{#2}\big(#3\big)}
\newcommand*{\iwin}[3][]{\iget[state]{#2} \in \iwinreg[#1]{#2}{#3}}
\newcommand*{\inowin}[3][]{\iget[state]{#2} \not\in \iwinreg[#1]{#2}{#3}}

\newcommand*{\restrictto}[2]{#1{\uparrow}#2}

\newcommand*{\selectlike}[3]{#1{\downarrow} #2{\scriptstyle(#3)}}
\newcommand*{\iselectlike}[3]{\selectlike{#1}{\iget[state]{#2}}{#3}}

\newcommand*{\genDE}[1]{\theta}%
\newcommand{\ivr}{\psi}

  \makeatletter
  \renewcommand{\iadjointSubst}[2]{%
    \useinterpretation{#2}%
    \edef\tmpadjointconst{{#1}^*_{\Interpretation@state}{\Interpretation@const}}%
    \iconcat[const=\tmpadjointconst]{#2}
  }
  \makeatother

\newcommand*{\usubstapp}[3]{{#2}^{#1}{#3}}%
\newcommand*{\usubstappp}[4]{{#3}^{#2}_{#1}{#4}}%

\usepackage{prettyref}
\newcommand{\rref}[2][]{\prettyref{#2}}
\newrefformat{sec}{Section\,\ref{#1}}
\newrefformat{def}{Def.\,\ref{#1}}
\newrefformat{thm}{Theorem\,\ref{#1}}
\newrefformat{prop}{Proposition\,\ref{#1}}
\newrefformat{lem}{Lemma\,\ref{#1}}
\newrefformat{cor}{Corollary\,\ref{#1}}
\newrefformat{ex}{Example\,\ref{#1}}
\newrefformat{tab}{Table\,\ref{#1}}
\newrefformat{fig}{Fig.\,\ref{#1}}
\newrefformat{case}{case\,\ref{#1}}
\newrefformat{app}{Appendix\,\ref{#1}}

\usepackage{fancyhdr}
\usepackage{ifpdf}
\ifpdf
\fi

\begin{document}

\newsavebox{\USarg}%
\sbox{\USarg}{$\boldsymbol{\cdot}$}

\newsavebox{\UScarg}%
\sbox{\UScarg}{$\boldsymbol{\_}$}

\newsavebox{\Lightningval}%
\sbox{\Lightningval}{$\scriptstyle\textcolor{vred}{\lightning}$}
\irlabel{clash|clash\usebox{\Lightningval}} %

\title{Uniform Substitution At One Fell Swoop}

\author{Andr\'e Platzer\thanks{
  Computer Science Department, Carnegie Mellon University, Pittsburgh, USA
\newline
  Fakult\"at f\"ur Informatik, Technische Universit\"at M\"unchen
}} 
\date{}

\maketitle
\allowdisplaybreaks
\thispagestyle{empty}

\maketitle

\begin{abstract}
Uniform substitution of function, predicate, program or game symbols is the core operation in parsimonious provers for hybrid systems and hybrid games.
By postponing soundness-critical admissibility checks, this paper introduces a uniform substitution mechanism that proceeds in a linear pass homomorphically along the formula.
Soundness is recovered using a simple variable condition at the replacements performed by the substitution.
The setting in this paper is that of differential hybrid games, in which discrete, continuous, and adversarial dynamics interact in differential game logic \dGL.
This paper proves soundness and completeness of one-pass uniform substitutions for \dGL.
\end{abstract}

\section{Introduction}

After a number of false starts on substitution \cite{HilbertAckermann28,HilbertBernays34,Quine34}, even by prominent logicians, did Church's \emph{uniform substitution} \cite[\S35,40]{Church_1956} provide a mechanism for substituting function and predicate symbols with terms and formulas in first-order logic.
Given a mechanism for applying a uniform substitution $\sigma$ to formulas $\phi$ with result denoted $\applyusubst{\sigma}{\phi}$ uniform substitutions are used with Church's proof rule:
\[
      \cinferenceRule[US|US]{uniform substitution}
      {\linferenceRule[formula]
        {\phi}
        {\applyusubst{\sigma}{\phi}}
      }{}
\]
Contrary to casual belief, quite some care is needed in the substitution process, even of only function symbols \cite{DBLP:journals/ndjfl/Schneider80}, in order to prevent replacing functions with terms that denote incompatible values in different places depending on which variables are being used in the replacements and in which formula contexts.
Due to their subtleties, there have even been passionate calls for banishing substitutions \cite{DBLP:journals/jsl/Henkin53} and using more schemata.
This paper moves in the opposite direction, making substitutions even more subtle, but also faster and, nevertheless, sound.

The biggest theoretical advantage of uniform substitutions is that they make instantiation explicit, so that proof calculi can use axioms (concrete object-level formulas) instead of axiom schemata (meta-level concepts standing for infinitely many formulas).
Their biggest practical advantage is that this avoidance of schemata enables parsimonious theorem prover implementations that only consist of copies of concrete formulas as axioms together with \emph{one} algorithm implementing the application of uniform substitutions (plus renaming).
Similar advantages exist for concrete axiomatic proof rules instead of rule schemata \cite{DBLP:journals/jar/Platzer17}.
This design obviates the need for algorithms that recognize all of the infinitely many instances of schemata and check all of their (sometimes pretty subtle) side conditions to soundly reject improper reasoning.
These practical advantages have first been demonstrated for hybrid systems \cite{DBLP:conf/cade/FultonMQVP15} and for hybrid games \cite{DBLP:conf/cade/Platzer18} proving, where uniform substitution led to significant reductions in soundness-critical size (down from 66000 to 1700 lines of code) or implementation time (down from months to minutes) compared to conventional prover implementations.

These uses of the uniform substitution principle required generalizations from first-order logic \cite{Church_1956} to differential dynamic logic \dL for hybrid systems \cite{DBLP:journals/jar/Platzer17} and differential game logic \dGL for hybrid games \cite{DBLP:conf/cade/Platzer18}, including substitutions of programs or games, respectively.
The presence of variables whose values change imperatively over time, and of differential equations \(\D{x}=\genDE{x}\) that cause intrinsic links of variables $x$ and their time-derivatives $\D{x}$, significantly complicate affairs compared to the simplicity of first-order logic \cite{Church_1956,DBLP:journals/ndjfl/Schneider80} and $\lambda$-calculus \cite{DBLP:journals/jsl/Church40}.
Pure $\lambda$-calculus has a single binder and rests on the three pillars of $\alpha$-conversions (for bound variables), $\beta$-reductions (by capture-avoiding substitutions), and $\eta$-conversions (versus free variables), which provide an elegant, deep, but solid foundation for functional programs (with similar observations for first-order logic).
Despite significant additional challenges,\footnote{%
The area of effect that an assignment to a variable has is non-computable and even a single occurrence of a variable may have to be both free and bound to ensure correctness.
Such overlap is an inherent consequence of change, which is an intrinsic feature of dynamical systems theory (the mathematics of change) and game theory (the mathematics of effects resulting from strategic interaction by player decisions).
}
just two elementary operations, nevertheless, suffice as a foundation for imperative programs and even hybrid games: bound renaming and uniform substitution (based on suitably generalized notions of free and bound variables).
Uniform substitutions generalize elegantly and in highly modular ways \cite{DBLP:journals/jar/Platzer17,DBLP:conf/cade/Platzer18}.
Much of the conceptual simplicity in the correctness arguments in these cases, however, came from the fact that Church-style uniform substitutions are applied by checking \emph{at each operator} admissibility, i.e., that no free variable be introduced into a context in which it is bound.
Such checks simplify correctness proofs, because they check each admissibility condition at every operator where they are necessary for soundness.
The resulting substitution mechanism is elegant but computationally suboptimal, because it repeatedly 
checks admissibility recursively again and again at every operator.
For example, applying a uniform substitution $\sigma$ checks at every sequential composition $\alpha;\beta$ again that the entire substitution $\sigma$ is admissible for the remainder $\beta$
compared to the bound variables of the result of having applied $\sigma$ to $\alpha$:
\begin{equation}
    \applyusubst{\sigma}{\usubstgroup{\alpha;\beta}} = (\applyusubst{\sigma}{\usubstgroup{\alpha}}; \applyusubst{\sigma}{\usubstgroup{\beta}}) \quad \text{if $\sigma$ is $\boundvarsdef{\applyusubst{\sigma}{\usubstgroup{\alpha}}}$-admissible for $\beta$}
  \label{eq:multipass-usubst-compose}
\end{equation}
where $\sigma$ is $U$-admissible for $\beta$ iff the free variables of the replacements for the part of $\sigma$ having function/predicate symbols that occur in $\beta$ do not intersect $U$, which, here, are the bound variables $\boundvarsdef{\applyusubst{\sigma}{\usubstgroup{\alpha}}}$ computed from the result of applying the substitution $\sigma$ to $\alpha$ \cite{DBLP:conf/cade/Platzer18}.
This mechanism is sound \cite{DBLP:journals/jar/Platzer17,DBLP:conf/cade/Platzer18}, even verified sound for hybrid systems in Isabelle/HOL and Coq \cite{DBLP:conf/cpp/BohrerRVVP17}, but computationally redundant due to its repeated substitution application and admissibility computations.

The point of this paper is to introduce a more liberal form of uniform substitution that \emph{substitutes at one fell swoop}, forgoing admissibility checks during the operators where they would be needed with a monadic computation of taboo sets to make up for that negligence by checking cumulative admissibility conditions locally only \emph{once} at each replacement that the uniform substitution application performs.
This \emph{one-pass uniform substitution} is computationally attractive, because it operates linearly in the output, which matters because uniform substitution is the dominant logical inference in uniform substitution provers \cite{DBLP:conf/cade/FultonMQVP15}.
The biggest challenge is, precisely, that correctness of substitution can no longer be justified for all operators where it is needed (because admissibility is no longer recursively checked at every operator).
The most important technical insight of this paper is that modularity of correctness arguments can be recovered, regardless, using a neighborhood semantics for taboos.
\iflongversion
\footnote{Interestingly, a one-pass formulation of uniform substitutions for \dL had originally been conjectured by the author, but abandoned because no soundness argument had been found. Simple modular soundness proofs were possible due to the simplicity of the prior Church-style uniform substitution formulations for \dL \cite{DBLP:journals/jar/Platzer17} and \dGL \cite{DBLP:conf/cade/Platzer18}.}
\fi
Another value of this paper is its straightforward completeness proof based on \cite{DBLP:journals/tocl/Platzer15,DBLP:journals/jar/Platzer17}.
Overall, the findings of this paper make it possible to verify hybrid games (and systems) with faster small soundness-critical prover cores than before \cite{DBLP:conf/cade/QueselP12,DBLP:conf/cade/Platzer18}, which, owing to their challenges, are the only two verification tools for hybrid games.
Uniform substitutions extend to differential games \cite{DBLP:journals/tams/ElliottK74,DBLP:journals/indianam/EvansSouganidis84}, where soundness is challenging \cite{DBLP:journals/tac/MitchellBT05}, leading to the first basis for a small prover core for differential hybrid games \cite{DBLP:journals/tocl/Platzer17}.
The accelerated proving primitives are of interest for other dynamic logics \cite{Harel_et_al_2000,KeYBook2016}.

\section{Preliminaries: Differential Game Logic}

This section recalls the basics of differential game logic \cite{DBLP:journals/tocl/Platzer15,DBLP:conf/cade/Platzer18}, the logic for specifying and verifying hybrid games of two players with differential equations.

\subsection{Syntax}

The set of all variables is $\allvars$, including for each variable $x$ a differential variable $\D{x}$ (e.g., for an ODE for $x$).
Higher-order differential variables $\D[2]{x}$ etc.\ are not used in this paper, so a finite set $\allvars$ suffices.
The terms $\theta$ of (differential-form) \dGL are polynomial terms with real-valued function symbols and \emph{differential terms} $\der{\theta}$ that are used to reduce reasoning about differential equations to reasoning about equations of differentials \cite{DBLP:journals/jar/Platzer17}.
Hybrid games $\alpha$ describe the permitted discrete and continuous actions by player Angel and player Demon.
Besides the operators of first-order logic of real arithmetic, \dGL formulas $\phi$ can be built using \(\ddiamond{\alpha}{\phi}\), which expresses that Angel has a winning strategy in the hybrid game $\alpha$ to reach the region satisfying \dGL formula $\phi$. Likewise, \(\dbox{\alpha}{\phi}\) expresses that Demon has a winning strategy in the hybrid game $\alpha$ to reach the region satisfying $\phi$.

\begin{definition}[Terms]
\emph{Terms} are defined by the following grammar
(with $\theta$,$\eta$, $\theta_1$,$\dots$,$\theta_k$ as terms, $x\in\allvars$ as variable, and $f$ as function symbol of arity $k$):
\[
  \theta,\eta ~\bebecomes~
  x
  \alternative
  f(\theta_1,\dots,\theta_k)
  \alternative
  \theta+\eta
  \alternative
  \theta\cdot\eta
  \alternative \der{\theta}
\]
\end{definition}

\begin{definition}[\dGL formulas] \label{def:dGL-formula}
The \emph{formulas of differential game logic {\dGL}} are defined by the following grammar (with $\phi,\psi$ as \dGL formulas, $p$ as predicate symbol of arity $k$, $\theta,\eta,\theta_i$ as terms, $x$ as variable, and $\alpha$ as hybrid game):
  \[
  \phi,\psi ~\bebecomes~
  \theta\geq\eta \alternative
  p(\theta_1,\dots,\theta_k) \alternative
  \lnot \phi \alternative
  \phi \land \psi \alternative
  \lexists{x}{\phi} \alternative 
  \ddiamond{\alpha}{\phi}
  \]
\end{definition}

The usual operators can be derived, e.g., \(\lforall{x}{\phi}\) is \(\lnot\lexists{x}{\lnot\phi}\) and similarly for $\limply,\lbisubjunct$ and truth $\ltrue$.
Existence of Demon's winning strategy in hybrid game $\alpha$ to achieve $\phi$ is expressed by the \dGL formula \(\dbox{\alpha}{\phi}\), which can be expressed indirectly as \(\lnot\ddiamond{\alpha}{\lnot\phi}\), thanks to the hybrid game determinacy theorem \cite[Thm.\,3.1]{DBLP:journals/tocl/Platzer15}.

\begin{definition}[Hybrid games] \label{def:dGL-HG}
The \emph{hybrid games of differential game logic {\dGL}} are defined by the following grammar (with $\alpha,\beta$ as hybrid games, $a$ as game symbol, $x$ as variable, $\theta$ as term, and $\ivr$ as \dGL formula):
\[
  \alpha,\beta ~\bebecomes~
  a\alternative
  \pupdate{\pumod{x}{\theta}}
  \alternative
  \pevolvein{\D{x}=\genDE{x}}{\ivr}
  \alternative
  \ptest{\ivr}
  \alternative
  \alpha\cup\beta
  \alternative
  \alpha;\beta
  \alternative
  \prepeat{\alpha}
  \alternative
  \pdual{\alpha}
\]
\end{definition}

The operator precedences make all unary operators, including modalities and quantifiers, bind stronger.
Just like the meaning of function and predicate symbols is subject to interpretation, the effect of game symbol $a$ is up to interpretation.
In contrast, the assignment game \(\pupdate{\pumod{x}{\theta}}\) has the specific effect of changing the value of variable $x$ to that of term $\theta$.
The differential equation game \(\pevolvein{\D{x}=\genDE{x}}{\ivr}\) allows Angel to choose how long she wants to follow the (vectorial) differential equation \(\D{x}=\genDE{x}\) for any real duration within the set of states where evolution domain constraint $\ivr$ is true.
Differential equation games with trivial \(\ivr=\ltrue\) are just written \(\pevolve{\D{x}=\genDE{x}}\).
The test game \(\ptest{\ivr}\) challenges Angel to satisfy formula $\ivr$, for if $\ivr$ is not true in the present state she loses the game prematurely.
The choice game \(\pchoice{\alpha}{\beta}\) allows Angel to choose if she wants to play game $\alpha$ or game $\beta$.
The sequential game \(\alpha;\beta\) will play game $\beta$ after game $\alpha$ terminates (unless a player prematurely lost the game while playing $\alpha$).
The repetition game \(\prepeat{\alpha}\) allows Angel to decide, after having played any number of $\alpha$ repetitions, whether she wants to play another round (but she cannot play forever).
Finally, the dual game \(\pdual{\alpha}\) will have both players switch sides: every choice that Angel had in $\alpha$ will go to Demon in \(\pdual{\alpha}\), and vice versa, while every condition that Angel needs to meet in $\alpha$ will be Demon's responsibility in \(\pdual{\alpha}\), and vice versa.

Substitutions are fundamental but subtle. 
For example, a substitution $\sigma$ that has the effect of replacing $f(x)$ with $x^2$ and $a(x)$ with $z\itimes y$ is unsound for the following formula while a substitution that replaces $a(x)$ with $z\itimes x^2$ would be fine:
\begin{equation}
    \linfer[clash]
    {\ddiamond{\pevolve{\D{x}=f(x)\syssep\D{y}=a(x)\itimes y}}{\,x\geq1} \lbisubjunct \ddiamond{\pevolve{\D{x}=f(x)}}{\,x\geq1}}
    {\ddiamond{\pevolve{\D{x}=x^2\syssep\D{y}={z\itimes y}\itimes y}}{\,x\geq1} \lbisubjunct \ddiamond{\pevolve{\D{x}=x^2}}{\,x\geq1}}
    \label{eq:firstbadsubstexample}
\end{equation}
The introduction of a new variable $z$ by the substitution $\sigma$ is acceptable, but, even if $y$ was already present previously, its introduction by $\sigma$ makes the inference unsound (e.g., when $x=y=1/z=1/2$), because this equates a system with a solution that is exponential in $y$ with a hyperbolic solution of more limited duration, even if both solutions are already hyperbolic of limited time from $x$.
By contrast, the use of the previously present variable $x$ to form \(\D{x}=x^2\) is fine.
The difference is that, unlike $z$, variable $y$ has a differential equation that changes the value of $y$ and, while $x$ also does, $f(x)$ and $a(x)$ may explicitly depend on $x$.
It is crucial to distinguish correct and incorrect substitutions in all cases.

\subsection{Semantics}

A \dfn{state} \(\iget[state]{\I}\) is a mapping from the set of all variables $\allvars$ to the reals $\reals$. 
The state \(\iget[state]{\imodif[state]{\I}{x}{r}}\) agrees with state \(\iget[state]{\I}\) except for variable $x$ whose value is $r\in\reals$ in \(\iget[state]{\imodif[state]{\I}{x}{r}}\).
The set of all states is denoted $\linterpretations{\Sigma}{\allvars}$ and the set of all its subsets is denoted \(\powerset{\linterpretations{\Sigma}{\allvars}}\).

The semantics of function, predicate, and game symbols is independent from the state.
They are interpreted by an \dfn{interpretation} $\iget[const]{\I}$ that maps each arity $k$ function symbol $f$ to a $k$-ary smooth function \(\iget[const]{\I}(f) : \reals^k\to\reals\), each arity $k$ predicate symbol $p$ to a $k$-ary relation \(\iget[const]{\I}(p) \subseteq \reals^k\), and each game symbol $a$ to a monotone \(\iget[const]{\I}(a) : \powerset{\linterpretations{\Sigma}{\allvars}} \to \powerset{\linterpretations{\Sigma}{\allvars}}\) where \(\iget[const]{\I}(a)(X) \subseteq \linterpretations{\Sigma}{\allvars}\) are the states from which Angel has a winning strategy to achieve $X\subseteq\linterpretations{Sigma}{V}$ in game $a$.
Differentials \(\der{\theta}\) have a differential-form semantics \cite{DBLP:journals/jar/Platzer17}: the sum of partial derivatives by all variables $x\in\allvars$ multiplied by the values of their associated differential variable $\D{x}$.

\begin{definition}[Semantics of terms] \label{def:dL-valuationTerm}
The \emph{semantics of a term} $\theta$ in interpretation $\iget[const]{\I}$ and state $\iget[state]{\I}\in\linterpretations{\Sigma}{V}$ is its value \m{\ivaluation{\I}{\theta}} in $\reals$.
It is defined inductively as
\begin{compactenum}
\item \m{\ivaluation{\I}{x} = \iget[state]{\I}(x)} for variable $x\in\allvars$
\item \(\ivaluation{\I}{f(\theta_1,\dots,\theta_k)} = \iget[const]{\I}(f)\big(\ivaluation{\I}{\theta_1},\dots,\ivaluation{\I}{\theta_k}\big)\) for function symbol $f$
\item \m{\ivaluation{\I}{\theta+\eta} = \ivaluation{\I}{\theta} + \ivaluation{\I}{\eta}}
\item \m{\ivaluation{\I}{\theta\cdot\eta} = \ivaluation{\I}{\theta} \cdot \ivaluation{\I}{\eta}}
\item
\m{\ivaluation{\I}{\der{\theta}}
= \sum_{x\in\allvars} \iget[state]{\I}(\D{x}) \Dp[x]{\ivaluation{\I}{\theta}}
}
for the differential \(\der{\theta}\) of $\theta$
\end{compactenum}
\end{definition}

\noindent
The semantics of differential game logic in interpretation $\iget[const]{\I}$ defines, for each formula $\phi$, the set of all states \(\imodel{\I}{\phi}\), in which $\phi$ is true.
Since hybrid games appear in \dGL formulas and vice versa, the semantics \(\iwinreg[\alpha]{\I}{X}\) of hybrid game $\alpha$ in interpretation $\iget[const]{\I}$ is defined by simultaneous induction as the set of all states from which Angel has a winning strategy in hybrid game $\alpha$ to achieve $X\subseteq\linterpretations{\Sigma}{\allvars}$.

\begin{definition}[\dGL semantics] \label{def:dGL-semantics}
The \emph{semantics of a \dGL formula} $\phi$ for each interpretation $\iget[const]{\I}$ with a corresponding set of states $\linterpretations{\Sigma}{\allvars}$ is the subset \m{\imodel{\I}{\phi}\subseteq\linterpretations{\Sigma}{\allvars}} of states in which $\phi$ is true.
It is defined inductively as follows
\begin{enumerate}
\item \(\imodel{\I}{\theta\geq\eta} = \{\iget[state]{\I} \in \linterpretations{\Sigma}{\allvars} \with \ivaluation{\I}{\theta}\geq\ivaluation{\I}{\eta}\}\)
\item \(\imodel{\I}{p(\theta_1,\dots,\theta_k)} = \{\iget[state]{\I} \in \linterpretations{\Sigma}{\allvars} \with (\ivaluation{\I}{\theta_1},\dots,\ivaluation{\I}{\theta_k})\in\iget[const]{\I}(p)\}\)
\item \(\imodel{\I}{\lnot\phi} = \scomplement{(\imodel{\I}{\phi})}\)
\(= \linterpretations{\Sigma}{\allvars} \setminus \imodel{\I}{\phi}\)
is the complement of \(\imodel{\I}{\phi}\)
\item \(\imodel{\I}{\phi\land\psi} = \imodel{\I}{\phi} \cap \imodel{\I}{\psi}\)
\item
{\def\Im{\imodif[state]{\I}{x}{r}}%
\(\imodel{\I}{\lexists{x}{\phi}} =  \{\iget[state]{\I} \in \linterpretations{\Sigma}{\allvars} \with \iget[state]{\Im} \in \imodel{\I}{\phi} ~\text{for some}~r\in\reals\}\)
}
\item \(\imodel{\I}{\ddiamond{\alpha}{\phi}} = \iwinreg[\alpha]{\I}{\imodel{\I}{\phi}}\)
\end{enumerate}
A \dGL formula $\phi$ is \emph{valid in $\iget[const]{\I}$}, written \m{\iget[const]{\I}\models{\phi}}, iff it is true in all states, i.e., \m{\imodel{\I}{\phi}=\linterpretations{\Sigma}{\allvars}}.
Formula $\phi$ is \emph{valid}, written \m{\entails\phi}, iff \m{\iget[const]{\I}\models{\phi}} for all interpretations $\iget[const]{\I}$.
\end{definition}

\begin{definition}[Semantics of hybrid games] \label{def:HG-semantics}
The \emph{semantics of a hybrid game} $\alpha$ for each interpretation $\iget[const]{\I}$ is a function \m{\iwinreg[\alpha]{\I}{{\cdot}}} that, for each set of states \m{X\subseteq\linterpretations{\Sigma}{\allvars}} as Angel's winning condition, gives the \emph{winning region}, i.e., the set of states \m{\iwinreg[\alpha]{\I}{X} \subseteq \linterpretations{\Sigma}{\allvars}} from which Angel has a winning strategy to achieve $X$ in $\alpha$ (whatever strategy Demon chooses).
It is defined inductively as follows
\begin{enumerate}
\item \(\iwinreg[a]{\I}{X} = \iget[const]{\I}(a)(X)\)
\item \(\iwinreg[\pupdate{\pumod{x}{\theta}}]{\I}{X} = \{\iget[state]{\I} \in \linterpretations{\Sigma}{\allvars} \with \modif{\iget[state]{\I}}{x}{\ivaluation{\I}{\theta}} \in X\}\)

\item \(\iwinreg[\pevolvein{\D{x}=\genDE{x}}{\ivr}]{\I}{X} = \{\iget[state]{\I} \in \linterpretations{\Sigma}{\allvars} \with 
      \iget[state]{\I}=\iget[state]{\Iff[0]}\) on $\scomplement{\{\D{x}\}}$ and \(\iget[state]{\Iff[r]}\in X\)
  for some function \m{\iget[flow]{\If}:[0,r]\to\linterpretations{\Sigma}{\allvars}} of some duration $r\in\reals$
  satisfying \m{\imodels{\If}{\D{x}=\genDE{x}\land\ivr}}$\}$
  \\
  where \m{\imodels{\If}{\D{x}=\genDE{x}\land\ivr}}
  iff \(\imodels{\Iff[\zeta]}{\D{x}=\genDE{x}\land\ivr}\)
  and \(\iget[state]{\Iff[0]}=\iget[state]{\Iff[\zeta]}\) on $\scomplement{\{x,\D{x}\}}$ 
  for all \(0{\leq}\zeta{\leq}r\)
  and \(\D[t]{\iget[state]{\Iff[t]}(x)}(\zeta)\) exists and equals \(\iget[state]{\Iff[\zeta]}(\D{x})\) for all \(0{\leq}\zeta{\leq}r\) if \m{r{>}0}.

\item \(\iwinreg[\ptest{\ivr}]{\I}{X} = \imodel{\I}{\ivr}\cap X\)
\item \(\iwinreg[\pchoice{\alpha}{\beta}]{\I}{X} = \iwinreg[\alpha]{\I}{X}\cup\iwinreg[\beta]{\I}{X}\)
\item \(\iwinreg[\alpha;\beta]{\I}{X} = \iwinreg[\alpha]{\I}{\iwinreg[\beta]{\I}{X}}\)
\item \(\iwinreg[\prepeat{\alpha}]{\I}{X} = \capfold\{Z\subseteq\linterpretations{\Sigma}{\allvars} \with X\cup\iwinreg[\alpha]{\I}{Z}\subseteq Z\}\)
which is a least fixpoint \cite{DBLP:journals/tocl/Platzer15}

\item \(\iwinreg[\pdual{\alpha}]{\I}{X} = \scomplement{(\iwinreg[\alpha]{\I}{\scomplement{X}})}\)
\end{enumerate}
\end{definition}

\noindent
Along \(\pevolvein{\D{x}=\genDE{x}}{\ivr}\), variables $x$ and $\D{x}$ enjoy an intrinsic link since they co-evolve.

\subsection{Static Semantics}

Sound uniform substitutions check free and bound occurrences of variables to prevent unsound replacements of expressions that might have incorrect values in the respective replacement contexts.
The whole point of this paper is to skip admissibility checks such as that in \rref{eq:multipass-usubst-compose}. Free (and, indirectly, bound) variables will still have to be consulted to tell apart acceptable from unsound occurrences.

Hybrid games even make it challenging to characterize free and bound variables. Both are definable based on whether or not their values affect the existence of winning strategies under variations of the winning conditions \cite{DBLP:conf/cade/Platzer18}.
The \emph{upward projection} \(\restrictto{X}{V}\) increases the winning condition $X\subseteq\linterpretations{\Sigma}{\allvars}$ from variables $V\subseteq\allvars$ to all states that are ``on $V$ like $X$'', i.e., similar on $V$ to states in $X$.
The \emph{downward projection} \(\iselectlike{X}{\I}{V}\) shrinks the winning condition $X$, fixing the values of state $\iget[state]{\I}$ on variables $V\subseteq\allvars$ to keep just those states of $X$ that agree with $\iget[state]{\I}$ on $V$.
\begin{definition} \label{def:projections}%
  The set \(\restrictto{X}{V} =
  \{ \iget[state]{\It} \in \linterpretations{\Sigma}{\allvars} \with \mexists{\iget[state]{\I}\in X}{\iget[state]{\I}=\iget[state]{\It} ~\text{on}~V}\} \supseteq X\)
  extends $X\subseteq\linterpretations{\Sigma}{\allvars}$ to the states that agree on $V\subseteq\allvars$ with some state in $X$ (written $\mexistsquantifier$).
  The set \(\iselectlike{X}{\I}{V} = \{ \iget[state]{\It}\in X \with \iget[state]{\I}=\iget[state]{\It} ~\text{on}~V\} \subseteq X\)
  selects state $\iget[state]{\I}$ on $V\subseteq\allvars$ in $X\subseteq\linterpretations{\Sigma}{\allvars}$.
\end{definition}

Projections make it possible to (\emph{semantically!}) define free and bound variables of hybrid games by expressing variable dependence and ignorance. Such semantic characterizations increase modularity and are used for the correctness of syntactic analyzes that compute supersets \cite[Sect.\,2.4]{DBLP:journals/jar/Platzer17}.
Variable $x$ is free in hybrid game $\alpha$ iff two states that only differ in the value of $x$ differ in membership in the winning region of $\alpha$ for some winning condition $\restrictto{X}{\scomplement{\{x\}}}$ that does not distinguish values of $x$. Variable $x$ is bound in hybrid game $\alpha$ iff it is in the winning region of $\alpha$ for some winning condition $X$ but not for the winning condition \(\iselectlike{X}{\I}{\{x\}}\) that limits the new value of $x$ to stay at its initial value $\iget[state]{\I}(x)$.

\begin{definition}[Static semantics] \label{def:static-semantics}%
The \emph{static semantics} defines the \emph{free variables}, which are all variables that the value of an expression depends on,
as well as \emph{bound variables}, $\boundvarsdef{\alpha}$, which can change their value during game $\alpha$, as:
\allowdisplaybreaks%
\[\begin{array}{ll}
  \freevarsdef{\theta} &= 
  \big\{x \in \allvars \with \mexists{\iget[const]{\I},\iget[state]{\I},\iget[state]{\IIalt}} \text{such that}~ \iget[state]{\I}=\iget[state]{\IIalt} ~\text{on}~\scomplement{\{x\}} ~ \text{and} ~ \ivaluation{\I}{\theta}\neq\ivaluation{\IIalt}{\theta}\big\}
  \\
  \freevarsdef{\phi} &= \big\{x \in \allvars \with \mexists{\iget[const]{\I},\iget[state]{\I},\iget[state]{\IIalt}} \text{such that}~\iget[state]{\I}=\iget[state]{\IIalt} ~\text{on}~\scomplement{\{x\}} ~ \text{and} ~ \imodels{\I}{\phi}\not\ni\iget[state]{\IIalt}\big\}
  \\
  \freevarsdef{\alpha} &= \big\{x \in \allvars \with \mexists{\iget[const]{\I},\iget[state]{\I},\iget[state]{\IIalt},X}
  \text{with}~\iget[state]{\I}=\iget[state]{\IIalt} ~\text{on}~\scomplement{\{x\}} ~\text{and}~
  \iget[state]{\I} \in \iwinreg[\alpha]{\I}{\restrictto{X}{\scomplement{\{x\}}}} \not\ni \iget[state]{\IIalt}\big\}
  \\
  \boundvarsdef{\alpha} &= \big\{x \in \allvars \with \mexists{\iget[const]{\I},\iget[state]{\I},X} \text{such that} ~ \iwinreg[\alpha]{\I}{X} \ni \iget[state]{\I} \not\in \iwinreg[\alpha]{\I}{\iselectlike{X}{\I}{\{x\}}} \big\}
\end{array}\]
\end{definition}

Beyond assignments, note complications with ODEs such as \rref{eq:firstbadsubstexample}, where, due to their nature as the solution of a fixpoint condition, the \emph{same} occurrences of variables are free, because they depend on their initial values, but they are also bound, because their values change along the ODE.
All occurrences of $x$ and $y$ but not $z$ on the right-hand side of \(\pevolve{\D{x}=x^2\syssep\D{y}={z\itimes x^2}\itimes y}\) and occurrences of $x,y,\D{x},\D{y}$ also after this ODE are bound, since they are affected by this change.
Variables $x,y,z$ but not $\D{x},\D{y}$ are free in this ODE.
The crucial need for overlap of free and bound variables is most obvious for ODEs, but also arises for loops, e.g., \(\prepeat{(\pupdate{\pumod{x}{x+1}};\pevolve{\D{x}=-x})}\).
If $x$ were not classified as free, its initial value could be overwritten incorrectly. If $x$ were not classified as bound, its initial value could be incorrectly copy-propagated across the loop. This also applies to the \emph{same} occurrence of $x$ in $x+1$ and $-x$, respectively. If it were not classified as a bound but a free occurrence, it could be incorrectly replaced by a term of the same initial value. If it were not classified as a free but a bound occurrence, it could, e.g., be boundly renamed, incorrectly losing its initial link.%
\footnote{These intricate variable relationships in games and the intrinsic link of $x$ and $\D{x}$ from ODEs significantly complicate substitutions beyond what is supported for first-order logic \cite{Church_1956,DBLP:journals/ndjfl/Schneider80}, $\lambda$-calculi \cite{DBLP:journals/jsl/Church40}, de Bruijn indices \cite{DeBruijn72}, or higher-order abstract syntax \cite{DBLP:conf/pldi/PfenningE88}.}

Coincidence lemmas \cite{DBLP:conf/cade/Platzer18} show truth-values of \dGL formulas only depend on their free variables (likewise for terms and hybrid games).
The bound effect lemma \cite{DBLP:conf/cade/Platzer18} shows only bound variables change their value when playing games.
Supersets satisfy the same lemmas, so corresponding \emph{syntactic} free and bound variable computations can be used correctly and are defined accordingly \cite{DBLP:journals/jar/Platzer17,DBLP:conf/cade/Platzer18}.
Since $\freevarsdef{}$ and $\boundvarsdef{}$ are the smallest such sets, no smaller sets can be correct, including, e.g., the usual definitions that classify occurrences mutually exclusively.

\begin{lemma}[Coincidence for terms \cite{DBLP:conf/cade/Platzer18}] \label{lem:coincidence-term}
  $\freevarsdef{\theta}$ is the smallest set with the coincidence property for $\theta$:
  If \(\iget[state]{\I}=\iget[state]{\Ialt}\) on $\freevarsdef{\theta}$,
  then \m{\ivaluation{\I}{\theta}=\ivaluation{\Ialt}{\theta}}.
\end{lemma}

\begin{figure}[tb]
\centering
\begin{tikzpicture}[scale=0.8, every node/.style={transform shape}]
  \useasboundingbox (-3.2,-1.6) rectangle (2,1.2); 
    \node[draw=vgreen,fill=vgreen!20,shape=rectangle,minimum height=2.7cm,minimum width=1cm] (Xalt) at (1.2,-0.2) {};
    \node[draw=vblue,fill=vblue!20,shape=ellipse,minimum height=1.3cm,minimum width=3cm,opacity=0.5] at (0.1,-0.8) {};
    \node at (1.2,0.2) {\m{\restrictto{X}{V}}};
    \node[winning condition,minimum width=0.5cm] (X) at (1.2,-0.8) {$X$};
    \node at (0,-1.2) {\(\iwinreg[\alpha]{\I}{X}\)};
    \node (v) at (-1,-0.8) {$\iget[state]{\I}$};
    \node (valt) at (-1,0.8) {$\iget[state]{\Ialt}$};
  \path
    (v) edge [similar state] node [left,align=right] {on $V\supseteq\freevarsdef{\alpha}$} (valt)
            edge [transition] node [above] {$\alpha$} (1.1,-0.8)
    (valt) edge [transition,transition exists] node [above] {$\alpha$}
            (1.1,0.8);
\end{tikzpicture}
\qquad
\begin{tikzpicture}[scale=0.8, every node/.style={transform shape}]
  \useasboundingbox (-3.2,-1.3) rectangle (2,1.3); 
    \node[winning condition,shape=ellipse,minimum height=2.5cm,minimum width=1.2cm] (X) at (1.2,0) {};
    \node[draw=vblue,fill=vblue!20,shape=ellipse,minimum height=1.5cm,minimum width=3.3cm,opacity=0.5,rotate=-30] at (0.1,0.1) {};
    \node[winning condition,draw=none,fill=none] at (1.2,0.8) {$X$};
    \node[winning condition,draw=vgreen,fill=vgreen!20,text=black,minimum width=0.5cm,inner sep=0pt] (selectX) at (1.2,-0.5) {$X{\downarrow}\iget[state]{\I}$};
    \node[rotate=-30] at (-0.2,-0.1) {\(\scriptstyle\iwinreg[\alpha]{\I}{\iselectlike{X}{\I}{\scomplement{\boundvarsdef{\alpha}}}}\)};
    \node (v) at (-1,0.8) {$\iget[state]{\I}$};
  \path
    (v) edge [transition] node [above,pos=0.7] {$\alpha$} (1,0.8);
  \path
    (v) edge [transition,dashed] node [above] {$\alpha$} (selectX);
\end{tikzpicture}
  \caption{Illustration of coincidence and bound effect properties of hybrid games}
  \label{fig:coincidence-bound}
\end{figure}
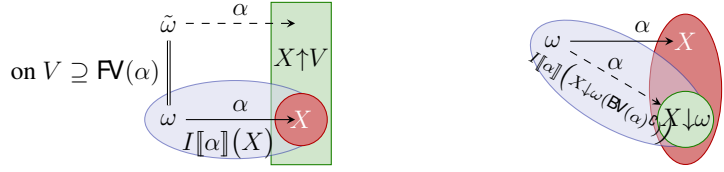

\begin{lemma}[Coincidence for formulas \cite{DBLP:conf/cade/Platzer18}] \label{lem:coincidence}
  $\freevarsdef{\phi}$ is the smallest set with the coincidence property for $\phi$:
  If \(\iget[state]{\I}=\iget[state]{\Ialt}\) on $\freevarsdef{\phi}$,
  then \m{\imodels{\I}{\phi}} iff \m{\imodels{\Ialt}{\phi}}.
\end{lemma}

\begin{lemma}[Coincidence for games \cite{DBLP:conf/cade/Platzer18}] \label{lem:coincidence-HG}
  $\freevarsdef{\alpha}$ is the smallest set with the coincidence property for $\alpha$:
  If \(\iget[state]{\I}=\iget[state]{\Ialt}\) on $V\supseteq\freevarsdef{\alpha}$,
  then \(\iwin[\alpha]{\I}{\restrictto{X}{V}}\)
  iff \(\iwin[\alpha]{\Ialt}{\restrictto{X}{V}}\);
  see \rref{fig:coincidence-bound}(left).
\end{lemma}

\begin{lemma}[Bound effect \cite{DBLP:conf/cade/Platzer18}] \label{lem:bound}
  $\boundvarsdef{\alpha}$ is the smallest set with the bound effect property for $\alpha$: 
  \(\iwin[\alpha]{\I}{X}\) iff \(\iwin[\alpha]{\I}{\iselectlike{X}{\I}{\scomplement{\boundvarsdef{\alpha}}}}\);
  see \rref{fig:coincidence-bound}(right).
\end{lemma}

The correctness of one-pass uniform substitution will become more transparent after defining when one state is a variation of another on a set of variables.
For a set $U\subseteq\allvars$, state $\iget[state]{\Ialt}$ is called a \emph{$U$-variation} of state $\iget[state]{\I}$ iff $\iget[state]{\Ialt}=\iget[state]{\I}$ on complement $\scomplement{U}$.
Variations satisfy properties of monotonicity and transitivity.
If $\iget[state]{\Ialt}$ is a $U$-variation of $\iget[state]{\I}$, then $\iget[state]{\Ialt}$ is a $V$-variation of $\iget[state]{\I}$ for all $V\supseteq U$.
If $\iget[state]{\Ialt}$ is a $U$-variation of $\iget[state]{\I}$ and $\iget[state]{\I}$ is a $V$-variation of $\iget[state]{\Iz}$, then $\iget[state]{\Ialt}$ is a $(U\cup V)$-variation of $\iget[state]{\Iz}$.
Coincidence lemmas say that the semantics is insensitive to variations of nonfree variables.
If $\iget[state]{\Ialt}$ is a $U$-variation of $\iget[state]{\I}$
and $\freevarsdef{\phi}\cap U=\emptyset$,
then \(\imodels{\I}{\phi}\) iff \(\imodels{\Ialt}{\phi}\).

\section{Uniform Substitution}

Uniform substitutions for \dGL affect terms, formulas, and games \cite{DBLP:conf/cade/Platzer18}.
A \dfn{uniform substitution} $\sigma$ is a mapping
from expressions of the
form \(f(\usarg)\) to terms $\applysubst{\sigma}{f(\usarg)}$,
from \(p(\usarg)\) to formulas $\applysubst{\sigma}{p(\usarg)}$,
and from game symbols \(a\) to hybrid games $\applysubst{\sigma}{a}$.
Here $\usarg$ is a reserved function symbol of arity 0 marking the position where the argument, e.g., argument $\theta$ to $p(\usarg)$ in formula $p(\theta)$, will end up in the replacement $\applysubst{\sigma}{p(\usarg)}$ used for $p(\theta)$.
Vectorial extensions would be accordingly for other arities $k\geq0$.

The key idea behind the new recursive one-pass application of uniform substitutions is that it simply applies $\sigma$ by na\"ive homomorphic recursion without checking any admissibility conditions along the way. But the mechanism makes up for that soundness-defying negligence by passing a cumulative set $U$ of taboo variables along the recursion that are then forbidden from being introduced free by $\sigma$ \emph{at the respective replacement} of function $f(\usarg)$ and predicate symbols $p(\usarg)$, respectively.
No corresponding condition is required at substitutions of game symbols $a$, since games already have unlimited access to and effect on the state.

\begin{figure}[tb]
  \renewcommand{\mequiv}{=}%
  \begin{displaymath}
    \begin{array}{@{}rcll@{}}
    \usubstapp{U}{\sigma}{\usubstgroup{x}} &=& x & \text{for variable $x\in\allvars$}\\
    \usubstapp{U}{\sigma}{\usubstgroup{f(\theta)}} &=& (\usubstapp{U}{\sigma}{f})(\usubstapp{U}{\sigma}{\theta})
  \mdefeq \usubstapp{\emptyset}{\{\usarg\mapsto\usubstapp{U}{\sigma}{\theta}\}}{\applysubst{\sigma}{f(\usarg)}} &
  \text{if}~\freevarsdef{\applysubst{\sigma}{f(\usarg)}} \cap U = \emptyset
  \\
  \usubstapp{U}{\sigma}{\usubstgroup{\theta+\eta}} &=& \usubstapp{U}{\sigma}{\theta} + \usubstapp{U}{\sigma}{\eta}
  \\
  \usubstapp{U}{\sigma}{\usubstgroup{\theta\cdot\eta}} &=& \usubstapp{U}{\sigma}{\theta} \cdot \usubstapp{U}{\sigma}{\eta}
  \\
  \usubstapp{U}{\sigma}{\usubstgroup{\der{\theta}}} &=& \der{\usubstapp{\allvars}{\sigma}{\theta}}
  \\
  \hline
  \usubstapp{U}{\sigma}{\usubstgroup{\theta\geq\eta}} &\mequiv& \usubstapp{U}{\sigma}{\theta} \geq \usubstapp{U}{\sigma}{\eta}\\
    \usubstapp{U}{\sigma}{\usubstgroup{p(\theta)}} &\mequiv& (\usubstapp{U}{\sigma}{p})(\usubstapp{U}{\sigma}{\theta})
  \mdefequiv \usubstapp{\emptyset}{\{\usarg\mapsto\usubstapp{U}{\sigma}{\theta}\}}{\applysubst{\sigma}{p(\usarg)}} &
  \text{if}~\freevarsdef{\applysubst{\sigma}{p(\usarg)}} \cap U = \emptyset
  \\
    \usubstapp{U}{\sigma}{\usubstgroup{\lnot\phi}} &\mequiv& \lnot\usubstapp{U}{\sigma}{\phi}\\
    \usubstapp{U}{\sigma}{\usubstgroup{\phi\land\psi}} &\mequiv& \usubstapp{U}{\sigma}{\phi} \land \usubstapp{U}{\sigma}{\psi}\\
    \usubstapp{U}{\sigma}{\usubstgroup{\lexists{x}{\phi}}} &\mequiv& \lexists{x}{\usubstapp{U\cup\{x\}}{\sigma}{\phi}} \\
    \usubstapp{U}{\sigma}{\usubstgroup{\ddiamond{\alpha}{\phi}}} &\mequiv& \ddiamond{\usubstappp{V}{U}{\sigma}{\alpha}}{\usubstapp{V}{\sigma}{\phi}}
    \\
  \hline
    \usubstappp{U\cup\boundvarsdef{\applysubst{\sigma}{a}}}{U}{\sigma}{\usubstgroup{a}} &\mequiv& \applysubst{\sigma}{a} &\text{for game symbol}~ a
    \\
    \usubstappp{U\cup\{x\}}{U}{\sigma}{\usubstgroup{\pupdate{\umod{x}{\theta}}}} &\mequiv& \pupdate{\umod{x}{\usubstapp{U}{\sigma}{\theta}}}\\
    \usubstappp{U\cup\{x,\D{x}\}}{U}{\sigma}{\usubstgroup{\pevolvein{\D{x}=\genDE{x}}{\ivr}}} &\mequiv&
    (\pevolvein{\D{x}=\usubstapp{U\cup\{x,\D{x}\}}{\sigma}{\genDE{x}}}{\usubstapp{U\cup\{x,\D{x}\}}{\sigma}{\ivr}}) \\
    \usubstappp{U}{U}{\sigma}{\usubstgroup{\ptest{\ivr}}} &\mequiv& \ptest{\usubstapp{U}{\sigma}{\ivr}}\\
    \usubstappp{V\cup W}{U}{\sigma}{\usubstgroup{\pchoice{\alpha}{\beta}}} &\mequiv& \pchoice{\usubstappp{V}{U}{\sigma}{\alpha}} {\usubstappp{W}{U}{\sigma}{\beta}}\\
    \usubstappp{W}{U}{\sigma}{\usubstgroup{\alpha;\beta}} &\mequiv& \usubstappp{V}{U}{\sigma}{\alpha}; \usubstappp{W}{V}{\sigma}{\beta}\\
    \usubstappp{V}{U}{\sigma}{\usubstgroup{\prepeat{\alpha}}} &\mequiv& \prepeat{(\usubstappp{V}{V}{\sigma}{\alpha})} &\text{where $\usubstappp{V}{U}{\sigma}{\alpha}$ is defined}\\
    \usubstappp{V}{U}{\sigma}{\usubstgroup{\pdual{\alpha}}} &\mequiv& \pdual{(\usubstappp{V}{U}{\sigma}{\alpha})} 
    \end{array}%
  \end{displaymath}%
  \caption{Recursive application of one-pass uniform substitution~$\sigma$ for taboo $U\subseteq\allvars$}%
  \index{substitution!uniform|textbf}%
  \label{fig:usubst-one}%
\end{figure}%

The result \(\usubstapp{U}{\sigma}{\phi}\) of \emph{applying uniform substitution $\sigma$ for taboo set $U\subseteq\allvars$ to a \dGL formula $\phi$} (or term $\theta$ or hybrid game $\alpha$, respectively) is defined in \rref{fig:usubst-one}.
For proof rule \irref{US}, the expression \m{\applyusubst{\sigma}{\phi}} is, then, defined to be \m{\usubstapp{\emptyset}{\sigma}{\phi}} without taboos.

The case for \(\lexists{x}{\phi}\) in \rref{fig:usubst-one} conjoins the variable $x$ to the taboo set in the homomorphic application of $\sigma$ to $\phi$, because any \emph{newly introduced} free uses of $x$ within that scope would refer to a different semantic value than outside that scope.
In addition to computing the substituted hybrid game \(\usubstappp{V}{U}{\sigma}{\alpha}\), the recursive application of one-pass uniform substitution $\sigma$ to hybrid game $\alpha$ under taboo set $U$ also performs an analysis that results in a new output taboo set $V$, written in subscript notation, that will be tabooed after this hybrid game.
Superscripts as inputs and subscripts as outputs follows static analysis notation and makes the $\alpha;\beta$ case reminiscent of Einstein's summation:
the output taboos $V$ of \(\usubstappp{V}{U}{\sigma}{\alpha}\) become the input taboos $V$ for \(\usubstappp{W}{V}{\sigma}{\beta}\), whose output $W$ is that of \(\usubstappp{W}{U}{\sigma}{\usubstgroup{\alpha;\beta}}\).
Similarly, the output taboos $V$ resulting from the uniform substitute \(\usubstappp{V}{U}{\sigma}{\alpha}\) of a hybrid game $\alpha$ become taboo during the uniform substitution application forming \(\usubstapp{V}{\sigma}{\phi}\) in the postcondition of a modality to build \(\usubstapp{U}{\sigma}{\usubstgroup{\ddiamond{\alpha}{\phi}}}\).

Repetitions \(\usubstappp{V}{U}{\sigma}{\usubstgroup{\prepeat{\alpha}}}\) are the only complication in \rref{fig:usubst-one}, where taboo $U$ would be too lax during the recursion, because earlier repetitions of $\alpha$ bind variables of $\alpha$ itself, so only the taboos $V$ obtained after one round \(\usubstappp{V}{U}{\sigma}{\alpha}\) are  correct input taboos for the loop body.
These two passes per loop are linear in the output when considering repetitions $\prepeat{\alpha}$ as their equivalent \(\pchoice{\ptest{\ltrue}}{\alpha;\prepeat{\alpha}}\) of double size.

Unlike in Church-style uniform substitution \cite{Church_1956,DBLP:journals/jar/Platzer17,DBLP:conf/cade/Platzer18}, attention is needed at the replacement sites of function and predicate symbols in order to make up for the neglected admissibility checks during all other operators.
The result \(\usubstapp{U}{\sigma}{\usubstgroup{p(\theta)}}\) of applying uniform substitution $\sigma$ with taboo $U$ to a predicate application \(p(\theta)\) is \emph{only} defined if the replacement \(\applysubst{\sigma}{p(\usarg)}\) for $p$ does not introduce free any tabooed variable,
i.e., \(\freevarsdef{\applysubst{\sigma}{p(\usarg)}} \cap U = \emptyset\).
Arguments are put in for placeholder $\usarg$ recursively by the taboo-free use of uniform substitution \({\{\usarg\mapsto\usubstapp{U}{\sigma}{\theta}\}}{}\), which replaces arity 0 function symbol $\usarg$ by $\usubstapp{U}{\sigma}{\theta}$.
Taboos $U$ are respected when forming (\emph{once!}) the uniform substitution to be used for argument $\usarg$, but empty taboos $\emptyset$ suffice when substituting the resulting \(\usubstapp{U}{\sigma}{\theta}\) for $\usarg$ in the replacement \(\applysubst{\sigma}{p(\usarg)}\) for $p$.

All variables $\allvars$ become taboos during uniform substitutions into differentials $\der{\theta}$, because any newly introduced occurrence of a variable $x$ would cause additional dependencies on its respective associated differential variable $\D{x}$.

If the conditions in \rref{fig:usubst-one} are not met, the substitution $\sigma$ is said to \dfn{clash} for taboo $U$ and its result \m{\usubstapp{U}{\sigma}{\phi}} is not defined and cannot be used.
\emph{All subsequent applications of uniform substitutions are required to be defined} (no clash).

Whether a substitution clashes is only checked once at each replacement, instead of also once per operator around it as in Church style from equation \rref{eq:multipass-usubst-compose}.
The free variables \(\freevarsdef{\applyusubst{\sigma}{p(\usarg)}}\) of each (function and) predicate symbol replacement are best stored with $\sigma$ to avoid repeated computation of free variables.

This inference would unsoundly equate linear solutions with exponential ones:
\[
\linfer[clash]
{\ddiamond{\pupdate{\pumod{v}{f}}}{p(v)} \lbisubjunct p(f)}
{\ddiamond{\pupdate{\pumod{v}{-x}}}{\dbox{\pevolve{\D{x}=v}}{\,x\geq0}} \lbisubjunct \dbox{\pevolve{\D{x}=-x}}{\,x\geq0}}
\]
Indeed, 
\(\sigma=\usubstlist{\usubstmod{p(\usarg)}{\dbox{\pevolve{\D{x}=\usarg}}{\,x\geq0}},\usubstmod{f}{-x}}\) clashes so rejects the above inference since the substitute $-x$ for $f$ has free variable $x$ that is taboo in the context \(\dbox{\pevolve{\D{x}=\usarg}}{\,x\geq0}\).
By contrast, a sound use of rule \irref{US}, despite its change in multiple binding contexts with
\(\sigma=\usubstlist{\usubstmod{p(\usarg)}{\dbox{\prepeat{(\pupdate{\pumod{x}{x+\usarg}};\pevolve{\D{x}=\usarg})}}{\,x+\usarg\geq0}},\usubstmod{f}{-v}}\), is:
\[
~~
\linfer[US]
{\ddiamond{\pupdate{\pumod{v}{f}}}{p(v)} \lbisubjunct p(f)}
{\ddiamond{\pupdate{\pumod{v}{-v}}}{\dbox{\prepeat{(\pupdate{\pumod{x}{x+v}};\pevolve{\D{x}=v})}}{\,x+v\geq0}} \lbisubjunct \dbox{\prepeat{(\pupdate{\pumod{x}{x-v}};\pevolve{\D{x}=-v})}}{\,x-v\geq0}}
\]
Uniform substitution accurately distinguishes such sound inferences from unsound ones even if the substitutions take effect deep down within a \dGL formula.
Uniform substitutions enable other syntactic transformations that require a solid understanding of variable occurrence patterns such as common subexpression elimination, for example, by using the above inference from right to left.

\subsection{Taboo Lemmas}

The only soundness-critical property of output taboos is that they correctly add bound variables and never forget variables that were already input taboos.

\begin{lemma}[Taboo set computation] \label{lem:usubst-taboos}
One-pass uniform substitution application monotonously computes taboos with correct bound variables for games:
\[
\text{if}~\usubstappp{V}{U}{\sigma}{\alpha}~\text{is defined,}
~\text{then}~
V \supseteq U\cup\boundvarsdef{\usubstappp{V}{U}{\sigma}{\alpha}}
\]
\end{lemma}
\begin{proofatend}
The proof is by direct structural induction on $\alpha$:
\begin{compactenum}
\item
  \(\usubstappp{U\cup\boundvarsdef{\applyusubst{\sigma}{a}}}{U}{\sigma}{\usubstgroup{a}} = \applyusubst{\sigma}{a}\),
  then \(V=U\cup\boundvarsdef{\applyusubst{\sigma}{a}}=U\cup\boundvarsdef{\usubstappp{V}{U}{\sigma}{a}}\) 
\item 
  \(\usubstappp{U\cup\{x\}}{U}{\sigma}{\usubstgroup{\pumod{x}{\theta}}} = (\pumod{x}{\usubstapp{U}{\sigma}{\theta}})\),
  then \(U\cup\{x\} \supseteq U\cup\boundvarsdef{\pumod{x}{\usubstapp{U}{\sigma}{\theta}}}\).

\item
  \(\usubstappp{U\cup\{x,\D{x}\}}{U}{\sigma}{\usubstgroup{\pevolvein{\D{x}=\genDE{x}}{\ivr}}} 
  = (\pevolvein{\D{x}=\usubstapp{U\cup\{x,\D{x}\}}{\sigma}{\genDE{x}}}
  {\usubstapp{U\cup\{x,\D{x}\}}{\sigma}{\ivr}})\),
  then it is, indeed, the case that \(U\cup\{x,\D{x}\} \supseteq 
  U\cup\boundvarsdef{\pevolvein{\D{x}=\usubstapp{U\cup\{x,\D{x}\}}{\sigma}{\genDE{x}}}
  {\usubstapp{U\cup\{x,\D{x}\}}{\sigma}{\ivr}}}\).

\item
  \(\usubstappp{U}{U}{\sigma}{\usubstgroup{\ptest{\ivr}}} =\, \ptest{\usubstapp{U}{\sigma}{\ivr}}\),
  then output $U$ is correct as \(\boundvarsdef{\ptest{\usubstapp{U}{\sigma}{\ivr}}}=\emptyset\).

\item
  \(\usubstappp{V\cup W}{U}{\sigma}{\usubstgroup{\pchoice{\alpha}{\beta}}} = \pchoice{\usubstappp{V}{U}{\sigma}{\alpha}}{\usubstappp{W}{U}{\sigma}{\beta}}\),
  then, by IH,
  \(V \supseteq U\cup\boundvarsdef{\usubstappp{V}{U}{\sigma}{\alpha}}\)
  and \(W \supseteq U\cup\boundvarsdef{\usubstappp{W}{U}{\sigma}{\beta}}\).
  Thus,
  \(V\cup W  \supseteq  U\cup\boundvarsdef{\usubstappp{V}{U}{\sigma}{\alpha}} \cup U\cup\boundvarsdef{\usubstappp{W}{U}{\sigma}{\beta}}
  \supseteq U\cup\boundvarsdef{\pchoice{\usubstappp{V}{U}{\sigma}{\alpha}}{\usubstappp{W}{U}{\sigma}{\beta}}}\).

\item
  \(\usubstappp{W}{U}{\sigma}{\usubstgroup{\alpha;\beta}} = \usubstappp{V}{U}{\sigma}{\alpha}; \usubstappp{W}{V}{\sigma}{\beta}\)
  then, by IH,
  \(V \supseteq U\cup\boundvarsdef{\usubstappp{V}{U}{\sigma}{\alpha}}\)
  and \(W \supseteq V\cup\boundvarsdef{\usubstappp{W}{V}{\sigma}{\beta}}\).
  Hence,
  \(W \supseteq U\cup\boundvarsdef{\usubstappp{V}{U}{\sigma}{\alpha}}\cup\boundvarsdef{\usubstappp{W}{V}{\sigma}{\beta}} \supseteq U\cup\boundvarsdef{\usubstappp{V}{U}{\sigma}{\alpha};\usubstappp{W}{V}{\sigma}{\beta}}\).

\item
  \(\usubstappp{V}{U}{\sigma}{\usubstgroup{\prepeat{\alpha}}} = \prepeat{(\usubstappp{V}{V}{\sigma}{\alpha})}\)
  if \(\usubstappp{V}{U}{\sigma}{\alpha}\) is defined.
  By IH on \(\usubstappp{V}{U}{\sigma}{\alpha}\),
  \(V \supseteq U\cup\boundvarsdef{\usubstappp{V}{U}{\sigma}{\alpha}}\).
  By IH on \(\usubstappp{V}{V}{\sigma}{\alpha}\),
  \(V \supseteq \boundvarsdef{\usubstappp{V}{V}{\sigma}{\alpha}}\).
  Hence,
  \(V \supseteq U\cup\boundvarsdef{\prepeat{(\usubstappp{V}{V}{\sigma}{\alpha})}}\) as \(\boundvarsdef{\alpha}\supseteq\boundvarsdef{\prepeat{\alpha}}\) for all games $\alpha$.

\item
  \(\usubstappp{V}{U}{\sigma}{\usubstgroup{\pdual{\alpha}}} = \pdual{(\usubstappp{V}{U}{\sigma}{\alpha})}\),
  then, by IH, 
  \(V \supseteq U\cup\boundvarsdef{\usubstappp{V}{U}{\sigma}{\alpha}}\).
  So,
  \(V \supseteq U\cup\boundvarsdef{\pdual{(\usubstappp{V}{U}{\sigma}{\alpha})}} \supseteq U\cup\boundvarsdef{\usubstappp{V}{U}{\sigma}{\alpha}}\).
\qedhere
\end{compactenum}
\end{proofatend}

Any superset of such taboo computations (or the free variable sets used in \rref{fig:usubst-one}) remains correct, just more conservative.
The change from input taboo $U$ to output taboo $V$ is a function of the hybrid game $\alpha$, justifying the construction of \(\usubstappp{V}{U}{\sigma}{\usubstgroup{\prepeat{\alpha}}}\):
if \(\usubstappp{V}{U}{\sigma}{\alpha}\) and \(\usubstappp{W}{V}{\sigma}{\alpha}\) are defined, then \(\usubstappp{V}{V}{\sigma}{\alpha}\) is defined and equal to \(\usubstappp{W}{V}{\sigma}{\alpha}\).
By \rref{lem:usubst-taboos}, no implementation of bound variables is needed when defining game symbols via 
\(\usubstappp{U\cup V}{U}{\sigma}{\usubstgroup{a}} = \applysubst{\sigma}{a}\) where \(\usubstappp{V}{\emptyset}{\usubstid}{\usubstgroup{\applysubst{\sigma}{a}}}\) with identity substitution $\usubstid$.
But bound variable computations speed up loops via
\(\usubstappp{V}{U}{\sigma}{\usubstgroup{\prepeat{\alpha}}} = \prepeat{(\usubstappp{V}{U\cup B}{\sigma}{\alpha})}\)
since \(B=\boundvarsdef{\usubstappp{M}{\emptyset}{\sigma}{\alpha}}\) can be computed and used correctly in one pass when \(U\cup B=V\).

\subsection{Uniform Substitution Lemmas}

Uniform substitutions are syntactic transformations on syntactic expressions.
Their semantic counterpart is the semantic transformation that maps an interpretation $\iget[const]{\I}$ and a state $\iget[state]{\I}$ to the adjoint interpretation $\iget[const]{\Ia}$ that changes the meaning of all symbols according to the syntactic substitution $\sigma$.
The interpretation \(\iget[const]{\imodif[const]{\I}{\,\usarg}{d}}\) agrees with $I$ except that function symbol $\usarg$ is interpreted as $d\in\reals$.%

\begin{definition}[Substitution adjoints] \label{def:adjointUsubst}
The \emph{adjoint} to substitution $\sigma$ is the operation that maps $\iportray{\I}$ to the \emph{adjoint} interpretation $\iget[const]{\Ia}$ in which the interpretation of each function symbol $f\ignore{\in\replacees{\sigma}}$, predicate symbol $p\ignore{\in\replacees{\sigma}}$, and game symbol $a\ignore{\in\replacees{\sigma}}$ are modified according to $\sigma$ (it is enough to consider those that $\sigma$ changes):
\allowdisplaybreaks%
\begin{align*}
  \iget[const]{\Ia}(f) &: \reals\to\reals;\, d\mapsto\ivaluation{\imodif[const]{\I}{\,\usarg}{d}}{\applysubst{\sigma}{f}(\usarg)}\\
  \iget[const]{\Ia}(p) &= \{d\in\reals \with \imodels{\imodif[const]{\I}{\,\usarg}{d}}{\applysubst{\sigma}{p}(\usarg)}\}
  \\
  \iget[const]{\Ia}(a) &: \powerset{\linterpretations{\Sigma}{\allvars}}\to\powerset{\linterpretations{\Sigma}{\allvars}};\, X\mapsto\iwinreg[\applysubst{\sigma}{a}]{\I}{X}
\end{align*}
\end{definition}

The uniform substitution lemmas below are key to the soundness and equate the syntactic effect that a uniform substitution $\sigma$ has on a syntactic expression in $\iportray{\I}$ with the semantic effect that the switch to the adjoint interpretation $\iget[const]{\Ia}$ has on the original expression.
The technical challenge compared to Church-style uniform substitution \cite{DBLP:journals/jar/Platzer17,DBLP:conf/cade/Platzer18} is that no admissibility conditions are checked at the game operators that need them, because the whole point of one-pass uniform substitution is that it homomorphically recurses in a linear complexity sweep by postponing admissibility checks.
All that happens during the substitution is that different taboo sets are passed along.
Yet, still, there is a crucial interplay of the particular taboos imposed henceforth at binding operators and the retroactive checking at function and predicate symbol replacement sites.

In order to soundly deal with the negligence in admissibility checking of one-pass uniform substitutions in a modular way, the main insight is that it is imperative to generalize the range of applicability of uniform substitution lemmas beyond the state $\iget[state]{\I}$ of original interest where the adjoint $\iget[const]{\Ia}$ was formed, and make them cover \emph{all} variations of states that are so similar that they might arise during soundness justifications.
By demanding more comprehensive care at replacement sites, soundness arguments make up for the temporary lapses in attention during all other operators.
This gives the uniform substitution algorithm broader liberties at binding operators, while simultaneously demanding broader compatibility in semantic neighborhoods on its parts.
Due to the recursive nature of function substitutions, the proof of the following result is by structural induction lexicographically on the structure of $\sigma$ and $\theta$, for all $U,\iget[state]{\Ie},\iget[state]{\I}$.

\begin{lemma}[Uniform substitution for terms] \label{lem:usubst-term}
The uniform substitution $\sigma$ for taboo $U\subseteq\allvars$ and its adjoint interpretation $\iget[const]{\Ia}$ for $\iportray{\I}$ have the same semantics on $U$-variations for all \emph{terms} $\theta$:
\[\text{for all $U$-variations $\iget[state]{\Ie}$ of $\iget[state]{\I}$:}~~\ivaluation{\Ie}{\usubstapp{U}{\sigma}{\theta}} = \ivaluation{\Iae}{\theta}\]
\end{lemma}
\begin{proofatend}
\def\Im{\vdLint[const=\usebox{\ImnI},state=\nu\oplus(x\mapsto\usebox{\Imnx})]}%
The proof is by structural induction lexicographically on the structure of $\sigma$ and of $\theta$, for all $U,\iget[state]{\Ie},\iget[state]{\I}$.
Fix any $U$-variation $\iget[state]{\Ie}$ of $\iget[state]{\I}$.
\begin{compactenum}
\item
  \(\ivaluation{\Ie}{\usubstapp{U}{\sigma}{x}} 
  = \ivaluation{\Ie}{x} = \iget[state]{\Ie}(x) \)
  = \(\ivaluation{\Iae}{x}\)
  since $\sigma$ changes no variables $x\in\allvars$

\item \label{case:usubst-funcapp0}
  Consider the arity zero case of function application, written $f()$ for emphasis:
  \(\ivaluation{\Ie}{\usubstapp{U}{\sigma}{\usubstgroup{f()}}}
  = \ivaluation{\Ie}{\applysubst{\sigma}{f()}}\),
  which, by \rref{lem:coincidence-term}, equals
  \(\ivaluation{\I}{\applysubst{\sigma}{f()}}
  = \iget[const]{\Iae}(f) = \ivaluation{\Iae}{f()}\),
  because $\iget[state]{\Ie}$ is a $U$-variation of $\iget[state]{\I}$
  and $\freevarsdef{\applysubst{\sigma}{f()}}\cap U=\emptyset$.

\item
  \newcommand{\Iea}{\iadjointSubst{\sigma}{\Ie}}%
  Let \(d\mdefeq\ivaluation{\Ie}{\usubstapp{U}{\sigma}{\theta}} 
  \overset{\text{IH}}{=} \ivaluation{\Iae}{\theta}\)
  by IH. 
  \(\ivaluation{\Ie}{\usubstapp{U}{\sigma}{\usubstgroup{f(\theta)}}}
  = \ivaluation{\Ie}{(\usubstapp{U}{\sigma}{f})\big(\usubstapp{U}{\sigma}{\theta}\big)}
  = \ivaluation{\Ie}{\usubstapp{\emptyset}{\{\usarg\mapsto\usubstapp{U}{\sigma}{\theta}\}}{\applysubst{\sigma}{f(\usarg)}}}\)
  \(\overset{\text{IH}}{=} \ivaluation{\Iminner}{\applysubst{\sigma}{f(\usarg)}}\),
    which equals \(\ivaluation{\Iiminner}{\applysubst{\sigma}{f(\usarg)}}=(\iget[const]{\Iae}(f))(d)\)
    by \rref{lem:coincidence-term} since 
    $\iget[state]{\Iea}$ is a $U$-variation of $\iget[state]{\I}$
    and $\freevarsdef{\applysubst{\sigma}{f(\usarg)}}\cap U=\emptyset$.
    Continuing,
    \((\iget[const]{\Iae}(f))(d)\)
  \(= (\iget[const]{\Iae}(f))(\ivaluation{\Iae}{\theta})
  = \ivaluation{\Iae}{f(\theta)}\).

  This proof used the induction hypothesis twice:
  once for \(\usubstapp{U}{\sigma}{\theta}\) on the smaller $\theta$ 
  and once for
  \(\usubstapp{\emptyset}{\{\usarg\mapsto\usubstapp{U}{\sigma}{\theta}\}}{\applysubst{\sigma}{f(\usarg)}}\)
  on the possibly bigger term
  \({\applysubst{\sigma}{f(\usarg)}}\)
  but the structurally simpler uniform substitution
  \(\applyusubst{\{\usarg\mapsto\usubstapp{U}{\sigma}{\theta}\}}{}\)
  that substitutes arity 0 symbol $\usarg$ instead of arity 1 function symbol $f$.
  For well-foundedness of the induction note that the $\usarg$ substitution only happens for function symbols $f$ with at least one argument $\theta$
  so not for $\usarg$ itself, which, as an arity zero function, is covered in \rref{case:usubst-funcapp0}.
  
\item
  \(\ivaluation{\Ie}{\usubstapp{U}{\sigma}{\usubstgroup{\theta+\eta}}}
  = \ivaluation{\Ie}{\usubstapp{U}{\sigma}{\theta} + \usubstapp{U}{\sigma}{\eta}}
  = \ivaluation{\Ie}{\usubstapp{U}{\sigma}{\theta}} + \ivaluation{\Ie}{\usubstapp{U}{\sigma}{\eta}}\)
  \(\overset{\text{IH}}{=} \ivaluation{\Iae}{\theta} + \ivaluation{\Iae}{\eta}\)
  \(= \ivaluation{\Iae}{\theta+\eta}\).
  The proof for multiplication \(\theta\cdot\eta\) is accordingly.

\item
\(
\ivaluation{\Ie}{\usubstapp{U}{\sigma}{\usubstgroup{\der{\theta}}}}
= \ivaluation{\Ie}{\der{\usubstapp{\allvars}{\sigma}{\theta}}}
= \sum_x \iget[state]{\Ie}(\D{x}) \itimes \Dp[x]{\ivaluation{\Ie}{\usubstapp{\allvars}{\sigma}{\theta}}}
\overset{\text{IH}}{=} \sum_x \iget[state]{\Ie}(\D{x}) \itimes \Dp[x]{\ivaluation{\Iae}{\theta}}\)
which is
\(\ivaluation{\Iae}{\der{\theta}}\)
since IH yields \(\ivaluation{\Ie}{\usubstapp{\allvars}{\sigma}{\theta}} = \ivaluation{\Iae}{\theta}\) for all states $\iget[state]{\Ie},\iget[state]{\I}$ (which are trivially $\allvars$-variations), including states used for partial derivatives.
\qedhere
\end{compactenum}
\end{proofatend}

Recall that all uniform substitutions are only defined when they meet the side conditions from \rref{fig:usubst-one}.
A mention such as \(\usubstapp{U}{\sigma}{\theta}\) in \rref{lem:usubst-term} implies that its side conditions during the application of $\sigma$ to $\theta$ with taboos $U$ are met.
Substitutions are antimonotone in taboos: If \(\usubstapp{U}{\sigma}{\theta}\) is defined, then \(\usubstapp{V}{\sigma}{\theta}\) is defined and equal to \(\usubstapp{U}{\sigma}{\theta}\) for all $V\subseteq U$ (accordingly for $\phi,\alpha$).
The more taboos a use of a substitution tolerates, the more broadly its adjoint generalizes to state variations.

The corresponding results for formulas and games are proved by simultaneous induction since formulas and games are defined by simultaneous induction, as games may occur in formulas and, vice versa.
The inductive proof is lexicographic over the structure of $\sigma$ and $\phi$ or $\alpha$, with a nested induction over the closure ordinals of the loop fixpoints, simultaneously for all $\nu,\omega,U,X$.

\begin{lemma}[Uniform substitution for formulas] \label{lem:usubst}
The uniform substitution $\sigma$ for taboo $U\subseteq\allvars$ and its adjoint interpretation $\iget[const]{\Ia}$ for $\iportray{\I}$ have the same semantics on $U$-variations for all \emph{formulas} $\phi$:
\[\text{for all $U$-variations $\iget[state]{\Ie}$ of $\iget[state]{\I}$:}~~\imodels{\Ie}{\usubstapp{U}{\sigma}{\phi}} ~\text{iff}~ \imodels{\Iae}{\phi}\]
\end{lemma}
\begin{proofatend}
The proof is by structural induction lexicographically on the structure of $\sigma$ and of $\phi$, with a simultaneous induction with the subsequent proof of \rref{lem:usubst-HG}, simultaneously for all $U,\iget[state]{\Ie},\iget[state]{\I}$.
Fix any $U$-variation $\iget[state]{\Ie}$ of $\iget[state]{\I}$.
\begin{compactenum}
\item
  \(\imodels{\Ie}{\usubstapp{U}{\sigma}{\usubstgroup{\theta\geq\eta}}}\)
  \(= \imodel{\Ie}{\usubstapp{U}{\sigma}{\theta} \geq \usubstapp{U}{\sigma}{\eta}}\)
  iff \(\ivaluation{\Ie}{\usubstapp{U}{\sigma}{\theta}} \geq \ivaluation{\Ie}{\usubstapp{U}{\sigma}{\eta}}\),
  by \rref{lem:usubst-term},
  iff \(\ivaluation{\Iae}{\theta} \geq \ivaluation{\Iae}{\eta}\)
  iff \(\imodels{\Iae}{\theta\geq\eta}\).

\item \label{case:usubst-predappskip}
  Consider a predicate symbol $q$ that is not substituted to anything else by $\sigma$:
  \(\imodels{\Ie}{\usubstapp{U}{\sigma}{\usubstgroup{q(\theta)}}}
  = \imodel{\Ie}{q(\usubstapp{U}{\sigma}{\theta})}\)
  iff \(\ivaluation{\Ie}{\usubstapp{U}{\sigma}{\theta}} \in \iget[const]{\Ie}(q)\)
  iff, by \rref{lem:usubst-term}, \(\ivaluation{\Iae}{\usubstapp{U}{\sigma}{\theta}} \in \iget[const]{\Ie}(q)\)
  iff \(\ivaluation{\Iae}{\usubstapp{U}{\sigma}{\theta}} \in \iget[const]{\Iae}(q)\)
  iff \(\imodels{\Iae}{q(\theta)}\)

\item
  Let \(d\mdefeq\ivaluation{\Ie}{\usubstapp{U}{\sigma}{\theta}} = \ivaluation{\Iae}{\theta}\)
  by \rref{lem:usubst-term} since $\iget[state]{\Ie}$ is a $U$-variation of $\iget[state]{\I}$.
  \(\imodels{\Ie}{\usubstapp{U}{\sigma}{\usubstgroup{p(\theta)}}}
  = \imodel{\Ie}{(\usubstapp{U}{\sigma}{p})\big(\usubstapp{U}{\sigma}{\theta}\big)}
  = \imodel{\Ie}{\usubstapp{\emptyset}{\{\usarg\mapsto\usubstapp{U}{\sigma}{\theta}\}}{\applysubst{\sigma}{p(\usarg)}}}\)
  iff \(\imodels{\Iminner}{\applysubst{\sigma}{p(\usarg)}}\) by IH,
  iff \(\imodels{\Iminnerorg}{\applysubst{\sigma}{p(\usarg)}}\) by \rref{lem:coincidence} as $\iget[state]{\Iminner}$ is a $U$-variation of $\iget[state]{\Iminnerorg}$ and \(\freevarsdef{\applyusubst{\sigma}{p(\usarg)}}\cap U=\emptyset\),
  iff \(d \in \iget[const]{\Ia}(p)\)
  iff \((\ivaluation{\Iae}{\theta}) \in \iget[const]{\Iae}(p)\)
  iff \(\imodels{\Iae}{p(\theta)}\).
  The IH for \(\usubstapp{\emptyset}{\{\usarg\mapsto\usubstapp{U}{\sigma}{\theta}\}}{\applysubst{\sigma}{p(\usarg)}}\)
  is used on the possibly bigger formula \({\applysubst{\sigma}{p(\usarg)}}\) but the structurally simpler uniform substitution \(\applyusubst{\{\usarg\mapsto\usubstapp{U}{\sigma}{\theta}\}}{}\) only substitutes function symbol $\usarg$ of arity zero, not predicates, thus is covered by \rref{case:usubst-predappskip}.

\item
  \(\imodels{\Ie}{\usubstapp{U}{\sigma}{\usubstgroup{\lnot\phi}}}
  = \imodel{\Ie}{\lnot\usubstapp{U}{\sigma}{\phi}}\)
  iff \(\inonmodels{\Ie}{\usubstapp{U}{\sigma}{\phi}}\)
  by IH
  iff \(\inonmodels{\Iae}{\phi}\)
  iff \(\imodels{\Iae}{\lnot\phi}\)

\item
  \(\imodels{\Ie}{\usubstapp{U}{\sigma}{\usubstgroup{\phi\land\psi}}}
  = \imodel{\Ie}{\usubstapp{U}{\sigma}{\phi} \land \usubstapp{U}{\sigma}{\psi}}
  = \imodel{\Ie}{\usubstapp{U}{\sigma}{\phi}} \cap \imodel{\Ie}{\usubstapp{U}{\sigma}{\psi}}\),
  by induction hypothesis,
  iff \(\imodels{\Iae}{\phi} \cap \imodel{\Iae}{\psi}
  = \imodel{\Iae}{\phi\land\psi}\)

\item
\def\Imd{\imodif[state]{\Ie}{x}{d}}%
\def\Iaemd{\imodif[state]{\Iae}{x}{d}}%
\def\Imda{\iadjointSubst{\sigma}{\Imd}}%
  \(\imodels{\Ie}{\usubstapp{U}{\sigma}{\usubstgroup{\lexists{x}{\phi}}}}
  = \imodel{\Ie}{\lexists{x}{\usubstapp{U\cup\{x\}}{\sigma}{\phi}}}\)
  iff for some $d$ \(\imodels{\Imd}{\usubstapp{U\cup\{x\}}{\sigma}{\phi}}\),
  so, by IH,
  iff (for some $d$ for any $(U\cup\{x\})$-variation $\iget[state]{\Iaemd}$ of $\iget[state]{\I}$: \(\imodels{\Iaemd}{\phi}\)),
  iff (for some $d$ for any $U$-variation $\iget[state]{\Ie}$ of $\iget[state]{\I}$: \(\imodels{\Iaemd}{\phi}\)),
  Thus, this is equivalent to
  \(\imodels{\Iae}{\lexists{x}{\phi}}\),
  because $\iget[state]{\Iae}$, indeed, is a $U$-variation of $\iget[state]{\I}$. 
  
\item
  \(\imodels{\Ie}{\usubstapp{U}{\sigma}{\usubstgroup{\ddiamond{\alpha}{\phi}}}}
  = \imodel{\Ie}{\ddiamond{\usubstappp{V}{U}{\sigma}{\alpha}}{\usubstapp{V}{\sigma}{\phi}}}\)
  = \(\iwinreg[{\usubstappp{V}{U}{\sigma}{\alpha}}]{\I}{\imodel{\Ie}{\usubstapp{V}{\sigma}{\phi}}}\)
  iff (by \rref{lem:bound})\\
  \(\iwin[\usubstappp{V}{U}{\sigma}{\alpha}]{\Ie}{\iselectlike{\imodel{\Ie}{\usubstapp{V}{\sigma}{\phi}}}{\Ie}{\scomplement{\boundvarsdef{\usubstappp{V}{U}{\sigma}{\alpha}}}}}\).
  Conversely:
  \(\imodels{\Iae}{\ddiamond{\alpha}{\phi}}\)
  = \(\iwinreg[\alpha]{\Ia}{\imodel{\Ia}{\phi}}\)
  iff (by \rref{lem:usubst-HG})
  \(\iwin[\usubstappp{V}{U}{\sigma}{\alpha}]{\Ie}{\imodel{\Ia}{\phi}}\)
  as \(\usubstappp{V}{U}{\sigma}{\alpha}\) is defined and $\iget[state]{\Ie}$ a $U$-variation of $\iget[state]{\I}$,
  iff (\rref{lem:bound})
  \(\iwin[\usubstappp{V}{U}{\sigma}{\alpha}]{\Ie}{\iselectlike{\imodel{\Ia}{\phi}}{\Ie}{\scomplement{\boundvarsdef{\usubstappp{V}{U}{\sigma}{\alpha}}}}}\).
  The conditions equate
  \[{\iselectlike{\imodel{\I}{\usubstapp{V}{\sigma}{\phi}}}{\Ie}{\scomplement{\boundvarsdef{\usubstappp{V}{U}{\sigma}{\alpha}}}}}
  = {\iselectlike{\imodel{\Ia}{\phi}}{\Ie}{\scomplement{\boundvarsdef{\usubstappp{V}{U}{\sigma}{\alpha}}}}}\]
  For this, consider any $\boundvarsdef{\usubstappp{V}{U}{\sigma}{\alpha}}$-variation $\iget[state]{\Iz}$ of $\iget[state]{\Ie}$ and show:
  \(\imodels{\Iaz}{\phi}\)
  iff \(\imodels{\Iz}{\usubstapp{V}{\sigma}{\phi}}\).
  By induction hypothesis,
  the latter is equivalent to
  \(\imodels{\Iaz}{\phi}\)
  when $\iget[state]{\Iaz}$ is a $V$-variation of $\iget[state]{\I}$,
  which it is, because $\iget[state]{\Iz}$ is a $\boundvarsdef{\usubstappp{V}{U}{\sigma}{\alpha}}$-variation of $\iget[state]{\Ie}$, which is, in turn, a $U$-variation of $\iget[state]{\I}$,
  so $\iget[state]{\Iz}$ is a $(U\cup\boundvarsdef{\usubstappp{V}{U}{\sigma}{\alpha}})$-variation of $\iget[state]{\I}$,
  hence also a $V$-variation, because $V\supseteq U\cup\boundvarsdef{\usubstappp{V}{U}{\sigma}{\alpha}}$ by \rref{lem:usubst-taboos}.
\qedhere
\end{compactenum}
\end{proofatend}

\begin{lemma}[Uniform substitution for games] \label{lem:usubst-HG}
The uniform substitution $\sigma$ for taboo $U\subseteq\allvars$ and its adjoint interpretation $\iget[const]{\Ia}$ for $\iportray{\I}$ have the same semantics on $U$-variations for all \emph{games} $\alpha$:
\[\text{for all $U$-variations $\iget[state]{\Ie}$ of $\iget[state]{\I}$:}~~\iwin[{\usubstappp{V}{U}{\sigma}{\alpha}}]{\Ie}{X} ~\text{iff}~ \iwin[\alpha]{\Iae}{X}\]
\end{lemma}
\begin{proofatend}
The proof is by lexicographic structural induction on $\sigma$ and $\alpha$, simultaneously with \rref{lem:usubst}, for all $U,\iget[state]{\Ie},\iget[state]{\I}$ and $X$.
Fix any $U$-variation $\iget[state]{\Ie}$ of $\iget[state]{\I}$.
\begin{compactenum}
\item
  \(\iwin[\usubstappp{U\cup\boundvarsdef{\applysubst{\sigma}{a}}}{U}{\sigma}{\usubstgroup{a}}]{\Ie}{X} = \iwinreg[\applysubst{\sigma}{a}]{\I}{X} = \iget[const]{\Ia}(a)(X) = \iwinreg[a]{\Ia}{X}\)
  for game $a$

\item 
  \(\iwin[\usubstappp{U\cup\{x\}}{U}{\sigma}{\usubstgroup{\pumod{x}{\theta}}}]{\Ie}{X}
  = \iwinreg[\pumod{x}{\usubstapp{U}{\sigma}{\theta}}]{\Ie}{X}\)
  iff \(X \ni \modif{\iget[state]{\Ie}}{x}{\ivaluation{\Ie}{\usubstapp{U}{\sigma}{\theta}}}\)
  = \(\modif{\iget[state]{\Ie}}{x}{\ivaluation{\Iae}{\theta}}\)
  by \rref{lem:usubst-term}, which is, thus, equivalent to
  \(\iwin[\pumod{x}{\theta}]{\Iae}{X}\).

\item
\newcommand{\Izeta}{\iconcat[state=\varphi(t)]{\I}}
\def\Izetaa{\iadjointSubst{\sigma}{\Izeta}}%
\newcommand{\Iazeta}{\iconcat[state=\varphi(t)]{\Ia}}
  \(\iwin[\usubstappp{U\cup\{x,\D{x}\}}{U}{\sigma}{\usubstgroup{\pevolvein{\D{x}=\genDE{x}}{\ivr}}}]{\Ie}{X}
  = \iwinreg[\pevolvein{\D{x}=\usubstapp{U\cup\{x,\D{x}\}}{\sigma}{\genDE{x}}}
  {\usubstapp{U\cup\{x,\D{x}\}}{\sigma}{\ivr}}]{\I}{X}\)
  iff \(\mexists{\varphi:[0,T]\to\linterpretations{\Sigma}{\allvars}}\)
  such that \(\varphi(0)=\iget[state]{\Ie}\) on $\scomplement{\{\D{x}\}}$, \(\varphi(T)\in X\) and for all \m{t\geq0}:
  \(\D[s]{\varphi(s)(x)}(t) = \ivaluation{\Izeta}{\usubstapp{U\cup\{x,\D{x}\}}{\sigma}{\genDE{x}}}
  = \ivaluation{\Iazeta}{\genDE{x}}\) by \rref{lem:usubst-term}
  and it also holds that
  \(\imodels{\Izeta}{\usubstapp{U\cup\{x,\D{x}\}}{\sigma}{\ivr}}\),
  which, by \rref{lem:usubst}, holds iff
  \(\imodels{\Iazeta}{\ivr}\).
  Here, \rref{lem:usubst-term} and~\ref{lem:usubst} are applicable,
  because $\iget[state]{\Izeta}$ is a $(U\cup\{x,\D{x}\})$-variation of $\iget[state]{\I}$, since
  $\iget[state]{\Izeta}$ is a $\{x,\D{x}\}$-variation of $\iget[state]{\Ie}$, which is a $U$-variation of $\iget[state]{\I}$.
  The latter two conditions are equivalent to  
  \(\iwin[\pevolvein{\D{x}=\genDE{x}}{\ivr}]{\Iae}{X}\).
  
\item
  \(\iwin[\usubstappp{U}{U}{\sigma}{\usubstgroup{\ptest{\ivr}}}]{\Ie}{X}
  = \iwinreg[\ptest{\usubstapp{U}{\sigma}{\ivr}}]{\Ie}{X} = \imodel{\Ie}{\usubstapp{U}{\sigma}{\ivr}} \cap X\)
  iff, by \rref{lem:usubst},
  \(\iget[state]{\Iae} \in \imodel{\Iae}{\ivr} \cap X\)
  \(= \iwinreg[\ptest{\ivr}]{\Iae}{X}\).

\item  
  \(\iwin[\usubstappp{V\cup W}{U}{\sigma}{\usubstgroup{\pchoice{\alpha}{\beta}}}]{\Ie}{X}
  = \iwinreg[\pchoice{\usubstappp{V}{U}{\sigma}{\alpha}}{\usubstappp{W}{U}{\sigma}{\beta}}]{\Ie}{X}\)
  = \(\iwinreg[\usubstappp{V}{U}{\sigma}{\alpha}]{\Ie}{X} \cup \iwinreg[\usubstappp{W}{U}{\sigma}{\beta}]{\Ie}{X}\),
  which, by IH, is equivalent to
  \(\iget[state]{\Iae} \in \iwinreg[\alpha]{\Iae}{X} \cup \iwinreg[\beta]{\Iae}{X}\)
  \(= \iwinreg[\pchoice{\alpha}{\beta}]{\Iae}{X}\).
  
\item
  \(\iwin[\usubstappp{W}{U}{\sigma}{\usubstgroup{\alpha;\beta}}]{\Ie}{X}
  = \iwinreg[\usubstappp{V}{U}{\sigma}{\alpha}; \usubstappp{W}{V}{\sigma}{\beta}]{\I}{X}\)
  = \(\iwinreg[\usubstappp{V}{U}{\sigma}{\alpha}]{\I}{\iwinreg[\usubstappp{W}{V}{\sigma}{\beta}]{\I}{X}}\) 
  iff, by \rref{lem:bound},
  \(\iwin[\usubstappp{V}{U}{\sigma}{\alpha}]{\Ie}{\iselectlike{\iwinreg[\usubstappp{W}{V}{\sigma}{\beta}]{\Ie}{X}}{\Ie}{\scomplement{\boundvarsdef{\usubstappp{V}{U}{\sigma}{\alpha}}}}}\).
  Starting conversely:\\
  \(\iwin[\alpha;\beta]{\Iae}{X}
  = \iwinreg[\alpha]{\Iae}{\iwinreg[\beta]{\Iae}{X}}\),
  iff, by IH, 
  \(\iwin[\usubstappp{V}{U}{\sigma}{\alpha}]{\Ie}{\iwinreg[\beta]{\Iae}{X}}\)
  iff, by Lem.\,\ref{lem:bound},
  \(\iwin[\usubstappp{V}{U}{\sigma}{\alpha}]{\Ie}{\iselectlike{\iwinreg[\beta]{\Iae}{X}}{\Ie}{\scomplement{\boundvarsdef{\usubstappp{V}{U}{\sigma}{\alpha}}}}}\).
  Both conditions equate:  
  \[
  {\iselectlike{\iwinreg[\usubstappp{W}{V}{\sigma}{\beta}]{\Ie}{X}}{\Ie}{\scomplement{\boundvarsdef{\usubstappp{V}{U}{\sigma}{\alpha}}}}}
  =
  {\iselectlike{\iwinreg[\beta]{\Ia}{X}}{\Ie}{\scomplement{\boundvarsdef{\usubstappp{V}{U}{\sigma}{\alpha}}}}}
  \]
  Consider any $\boundvarsdef{\usubstappp{V}{U}{\sigma}{\alpha}}$-variation $\iget[state]{\Iz}$ of $\iget[state]{\Ie}$ to show:
  \(\iwin[\usubstappp{W}{V}{\sigma}{\beta}]{\Iz}{X}\)
  iff \(\iwin[\beta]{\Iaz}{X}\).
  This holds by IH, 
  because $\iget[state]{\Iz}$ is a $V$-variation of $\iget[state]{\I}$:
  $\iget[state]{\Iz}$ is a $\boundvarsdef{\usubstappp{V}{U}{\sigma}{\alpha}}$-variation of $\iget[state]{\Ie}$,
  which, in turn, is a $U$-variation of $\iget[state]{\I}$,
  so $\iget[state]{\Iz}$ is a $(U\cup\boundvarsdef{\usubstappp{V}{U}{\sigma}{\alpha}})$-variation of $\iget[state]{\I}$,
  hence a $V$-variation by \rref{lem:usubst-taboos}. 
  
\item
{%
\newcommand{\inflop}[2][]{\tau^{#1}(#2)}%
\newcommand{\oinflop}[2][]{\varrho^{#1}(#2)}%
  The case \(\iwin[\usubstappp{V}{U}{\sigma}{\usubstgroup{\prepeat{\alpha}}}]{\Ie}{X}
  = \iwinreg[\prepeat{(\usubstappp{V}{V}{\sigma}{\alpha})}]{\I}{X}\)
  (when \(\usubstappp{U}{V}{\sigma}{\alpha}\) is defined)
  uses an equivalent inflationary fixpoint formulation \cite[Thm.\,3.5]{DBLP:journals/tocl/Platzer15}:%
  \allowdisplaybreaks%
  \begin{align*}
    \inflop[0]{X} &\mdefeq X\\
    \inflop[\kappa+1]{X} &\mdefeq X \cup \iwinreg[\usubstappp{V}{V}{\sigma}{\alpha}]{\I}{\inflop[\kappa]{X}} && \kappa+1~\text{a successor ordinal}\\
    \inflop[\lambda]{X} &\mdefeq \cupfold_{\kappa<\lambda} \inflop[\kappa]{X} && \lambda\neq0~\text{a limit ordinal}
  \intertext{%
  where the union \(\inflop[\infty]{X} = \cupfold_{\kappa<\infty} \inflop[\kappa]{X}\) over all ordinals is \(\iwinreg[\prepeat{(\usubstappp{V}{V}{\sigma}{\alpha})}]{\I}{X}\).
  A similar fixpoint works for the other side \(\iwinreg[\prepeat{\alpha}]{\Ia}{X} = \oinflop[\infty]{X}\) where:}%
\iflongversion
    \oinflop[0]{X} &\mdefeq X\\
\else\\[-1.7\baselineskip]
\fi
    \oinflop[\kappa+1]{X} &\mdefeq X \cup \iwinreg[\alpha]{\Ia}{\oinflop[\kappa]{X}} && \kappa+1~\text{a successor ordinal}\\
\iflongversion
    \oinflop[\lambda]{X} &\mdefeq \cupfold_{\kappa<\lambda} \oinflop[\kappa]{X} && \lambda\neq0~\text{a limit ordinal}
\else\\[-2\baselineskip]
\fi
  \end{align*}%
  The equivalence
  \(\iwin[\usubstappp{V}{V}{\sigma}{\usubstgroup{\prepeat{\alpha}}}]{\Ie}{X} = \inflop[\infty]{X}\)
  iff
  \(\iwin[\prepeat{\alpha}]{\Iae}{X} = \oinflop[\infty]{X}\)
  for all $U$-variations $\iget[state]{\Ie}$ of $\iget[state]{\I}$
  follows, with $V\supseteq U$ by \rref{lem:usubst-taboos}, from proving:
  \[
  \text{for all}~\kappa~
  \text{and all}~X
  ~\text{and all $V$-variations}~\iget[state]{\Ie} ~\text{of}~\iget[state]{\I}:
  ~
  \iget[state]{\Ie} \in \inflop[\kappa]{X} ~\text{iff}~ \iget[state]{\Ie} \in \oinflop[\kappa]{X}
  \]
  This is proved by induction on ordinal $\kappa$ (0, limit ordinal $\lambda\neq0$, or successor):
  \begin{enumerate}
  \item[$\kappa=0$:]
  \(\iget[state]{\Ie} \in \inflop[0]{X} ~\text{iff}~ \iget[state]{\Ie} \in \oinflop[0]{X}\), because both sets equal $X$.
  
  \item[$\lambda$:]
  \(\iget[state]{\Ie} \in \inflop[\lambda]{X} 
  = \cupfold_{\kappa<\lambda} \inflop[\kappa]{X}\)
  iff there is a $\kappa<\lambda$ such that
  \(\iget[state]{\Ie} \in \inflop[\kappa]{X}\)
  iff, by IH,
  \(\iget[state]{\Ie} \in \oinflop[\kappa]{X}\) for some $\kappa<\lambda$,
  iff
  \(\iget[state]{\Ie} \in \cupfold_{\kappa<\lambda} \oinflop[\kappa]{X}
  = \oinflop[\lambda]{X}\).
  
  \item[$\kappa+1$:]
  \(\iget[state]{\Ie} \in \inflop[\kappa+1]{X}
  = X \cup \iwinreg[\usubstappp{V}{V}{\sigma}{\alpha}]{\Ie}{\inflop[\kappa]{X}}\),
  is equivalent, by \rref{lem:bound}, to
  \(\iget[state]{\Ie} \in X \cup \iwinreg[\usubstappp{V}{V}{\sigma}{\alpha}]{\Ie}{\iselectlike{\inflop[\kappa]{X}}{\Ie}{\scomplement{\boundvarsdef{\usubstappp{V}{V}{\sigma}{\alpha}}}}}\).
  Conversely,
  \(\iget[state]{\Ie} \in \oinflop[\kappa+1]{X}
  = X \cup \iwinreg[\alpha]{\Iae}{\oinflop[\kappa]{X}}\)
  iff, by IH on $\alpha$,
  \(\iget[state]{\Ie} \in X \cup \iwinreg[\usubstappp{V}{V}{\sigma}{\alpha}]{\Ie}{\oinflop[\kappa]{X}}\)
  for any $V$-variations $\iget[state]{\Ie}$ of $\iget[state]{\I}$,
  iff, by \rref{lem:bound},
  \(\iget[state]{\Ie} \in X \cup \iwinreg[\usubstappp{V}{V}{\sigma}{\alpha}]{\Ie}{\iselectlike{\oinflop[\kappa]{X}}{\Ie}{\scomplement{\boundvarsdef{\usubstappp{V}{V}{\sigma}{\alpha}}}}}\).
  Now
  \(
  {\iselectlike{\inflop[\kappa]{X}}{\Ie}{\scomplement{\boundvarsdef{\usubstappp{V}{V}{\sigma}{\alpha}}}}}
  =
  {\iselectlike{\oinflop[\kappa]{X}}{\Ie}{\scomplement{\boundvarsdef{\usubstappp{V}{V}{\sigma}{\alpha}}}}}
  \)
  holds as follows.
  Consider any $\boundvarsdef{\usubstappp{V}{V}{\sigma}{\alpha}}$-variation $\iget[state]{\Iz}$ of $\iget[state]{\Ie}$ and show:
  \(\iget[state]{\Iz} \in \inflop[\kappa]{X}\)
  iff
  \(\iget[state]{\Iz} \in \oinflop[\kappa]{X}\),
  which is by IH on $\kappa<\kappa+1$, as $\iget[state]{\Iz}$ is a $V$-variation of $\iget[state]{\I}$:
  $\iget[state]{\Iz}$ is a $\boundvarsdef{\usubstappp{V}{V}{\sigma}{\alpha}}$-variation of $\iget[state]{\Ie}$, so by $V\supseteq\boundvarsdef{\usubstappp{V}{V}{\sigma}{\alpha}}$ from \rref{lem:usubst-taboos}, $\iget[state]{\Iz}$ is a $V$-variation of $\iget[state]{\Ie}$, which, in turn, is a $U$-variation of $\iget[state]{\I}$, hence, by $V\supseteq U$ from \rref{lem:usubst-taboos} as \(\usubstappp{V}{U}{\sigma}{\alpha}\) is defined, also a $V$-variation of $\iget[state]{\I}$, so $\iget[state]{\Iz}$ itself is a $V$-variation of $\iget[state]{\I}$.
  \end{enumerate}%
}%

\item
  \(\iwin[\usubstappp{V}{U}{\sigma}{\usubstgroup{\pdual{\alpha}}}]{\Ie}{X}
  = \iwinreg[\pdual{(\usubstappp{V}{U}{\sigma}{\alpha})}]{\Ie}{X}
  = \scomplement{\big(\iwinreg[\usubstappp{V}{U}{\sigma}{\alpha}]{\Ie}{\scomplement{X}}\big)}\)
  iff \(\inowin[\usubstappp{V}{U}{\sigma}{\alpha}]{\Ie}{\scomplement{X}}\),
  iff, by IH, \(\inowin[\alpha]{\Iae}{\scomplement{X}}\),
  iff
  \(\iget[state]{\Iae} \in \scomplement{\big(\iwinreg[\alpha]{\Iae}{\scomplement{X}}\big)}
  = \iwinreg[\pdual{\alpha}]{\Iae}{X}\).
  \qedhere
\end{compactenum}
\end{proofatend}

\subsection{Soundness}

With the uniform substitution lemmas having established the crucial equivalence of syntactic substitution and adjoint interpretation, the soundness of uniform substitution uses in proofs is now immediate.
The notation \m{\applyusubst{\sigma}{\phi}} in proof rule \irref{US} is short for \m{\usubstapp{\emptyset}{\sigma}{\phi}}, so the result of applying $\sigma$ to $\phi$ without taboos (more taboos may still arise during the substitution application), and only defined if \m{\usubstapp{\emptyset}{\sigma}{\phi}} is.
A proof rule is \emph{sound} when its conclusion is valid if all its premises are valid.

\begin{theorem}[Soundness of uniform substitution]%
  \label{thm:usubst-sound}%
  Proof rule \irref{US} is sound.%
  {\upshape\[
  \cinferenceRuleQuote{US}
  \]}%
\end{theorem}
\begin{proof}
\def\Ia{\iadjointSubst{\sigma}{\I}}%
Let the premise $\phi$ of \irref{US} be valid, i.e., \m{\imodels{\I}{\phi}} for all interpretations $\iget[const]{\I}$ and states $\iget[state]{\I}$.
To show that the conclusion is valid, consider any $\iget[const]{\I}$ and state $\iget[state]{\I}$ and show
\(\imodels{\I}{\applyusubst{\sigma}{\phi}} = \imodel{\I}{\usubstapp{\emptyset}{\sigma}{\phi}}\).
By \rref{lem:usubst}, \(\imodels{\I}{\usubstapp{\emptyset}{\sigma}{\phi}}\) iff \(\imodels{\Ia}{\phi}\).
Now \(\imodels{\Ia}{\phi}\) holds, because \(\imodels{\I}{\phi}\) for all $\iportray{\I}$, including $\iportray{\Ia}$, by premise.
\qedhere
\end{proof}

\rref{thm:usubst-sound} is all it takes to soundly instantiate concrete axioms.
Uniform substitutions can instantiate whole inferences \cite{DBLP:journals/jar/Platzer17}, which makes it possible to avoid proof rule schemata by instantiating axiomatic proof rules consisting of pairs of concrete formulas.
This enables uniformly substituting premises and conclusions of entire proofs of \emph{locally sound} inferences, i.e., those whose conclusion is valid in any interpretation that all their premises are valid in.

\begin{theorem}[Soundness of uniform substitution of rules] \label{thm:usubst-rule}
  All uniform substitution instances for taboo $\allvars$ of locally sound inferences are locally sound:%
  \[
\linfer
{\phi_1 \quad \dots \quad \phi_n}
{\psi}
~~\text{locally sound}\quad\text{implies}\quad
\linfer%
{\usubstapp{\allvars}{\sigma}{\phi_1} \quad \dots \quad \usubstapp{\allvars}{\sigma}{\phi_n}}
{\usubstapp{\allvars}{\sigma}{\psi}}
~~\text{locally sound}
\irlabel{USR|USR}
  \]
\end{theorem}
\begin{proof}
\def\locproof{\mathcal{D}}%
Fix any state $\iget[state]{\I}$.
Let $\locproof$ be the locally sound inference on the left and $\applyusubst{\sigma}{\locproof}$ the substituted inference on the right.
To prove $\applyusubst{\sigma}{\locproof}$ locally sound, consider any interpretation $\iget[const]{\I}$ in which all premises of $\applyusubst{\sigma}{\locproof}$ are valid, i.e.,
\(\iget[const]{\Ie}\models{\usubstapp{\allvars}{\sigma}{\phi_j}}\) for all $j$, i.e.,
\(\imodels{\Ie}{\usubstapp{\allvars}{\sigma}{\phi_j}}\) for all $\iget[state]{\Ie}$ and $j$.
By \rref{lem:usubst}, \(\imodels{\Ie}{\usubstapp{\allvars}{\sigma}{\phi_j}}\) is equivalent to
\(\imodels{\Iae}{\phi_j}\)
(since $\iget[state]{\Ie}$ is a $\allvars$-variation of $\iget[state]{\I}$),
which also holds for all $\iget[state]{\Ie}$ and $j$.

Consequently, all premises of $\locproof$ are valid in the same interpretation $\iget[const]{\Iae}$, i.e. \(\iget[const]{\Iae}\models{\phi_j}\) for all $j$.
Thus, \(\iget[const]{\Iae}\models{\psi}\) by local soundness of $\locproof$.
That is, \(\imodels{\Iae}{\psi}\) for all $\iget[state]{\Iae}$.
By \rref{lem:usubst}, \(\imodels{\Iae}{\psi}\) is equivalent to \(\imodels{\Ie}{\usubstapp{\allvars}{\sigma}{\psi}}\) (since $\iget[state]{\Ie}$ trivially is a $\allvars$-variation of $\iget[state]{\I}$),
which continues to hold for all $\iget[state]{\Ie}$.
Thus, \(\iget[const]{\Ie}\models{\usubstapp{\allvars}{\sigma}{\psi}}\), i.e., the conclusion of $\applyusubst{\sigma}{\locproof}$ is valid in $\iget[const]{\I}$, hence $\applyusubst{\sigma}{\locproof}$ is locally sound.
\qedhere
\end{proof}

\noindent
\irref{USR} marks the use of \rref{thm:usubst-rule} in proofs. 
If $n=0$ (so $\psi$ has a proof), \irref{USR} preserves local soundness for taboo-free \(\usubstapp{\emptyset}{\sigma}{\psi}\) instead of \(\usubstapp{\allvars}{\sigma}{\psi}\), as \irref{US} proves \(\usubstapp{\emptyset}{\sigma}{\psi}\) from the provable $\psi$ and soundness is equivalent to local soundness for $n=0$.

\subsection{Completeness}

Soundness is the property that every formula with a proof is valid. This is the most important consideration for something as fundamental as a uniform substitution mechanism.
But the converse question of completeness, i.e., that every valid formula has a proof, is of interest as well, especially given the fact that one-pass uniform substitutions check differently for soundness during the substitution application, which had better not lose otherwise perfectly valid proofs.

Completeness is proved in an easy modular style based on all the nontrivial findings summarized in schematic relative completeness results, first for schematic \dGL \cite[Thm.\,4.5]{DBLP:journals/tocl/Platzer15}, and then for a uniform substitution formulation of \dL \cite[Thm.\,40]{DBLP:journals/jar/Platzer17}.
The combination of both schematic completeness results makes it fairly easy to lift completeness to the setting in this paper.
The challenge is to show that all instances of axiom schemata that are used for \dGL's schematic relative completeness result are provable by one-pass uniform substitution.

A \dGL formula $\phi$ is called \emph{surjective} iff rule \irref{US} can instantiate $\phi$ to any of its axiom schema instances, i.e., those formulas that are obtained by just replacing game symbols $a$ uniformly by any game, etc.
An axiomatic rule is called \emph{surjective} iff \irref{USR} of \rref{thm:usubst-rule} can instantiate it to any of its proof rule schema instances.

\begin{lemma}[Surjective axioms] \label{lem:surjectiveaxiom}
  If $\phi$ is a \dGL formula that is built only from  game symbols but no function or predicate symbols, then $\phi$ is surjective.
  Axiomatic rules consisting of surjective \dGL formulas are surjective.
\end{lemma}
\begin{proofatend}
Let $\tilde{\phi}$ be the desired instance of schema $\phi$. So, $\tilde{\phi}$ is obtained from $\phi$ by uniformly replacing each game symbol $a$ by some hybrid game, na\"ively but consistently (same replacement for $a$ in all places).
A straightforward structural induction on $\phi$ proves that there is a uniform substitution $\sigma$ such that \(\usubstapp{\allvars}{\sigma}{\phi} = \tilde{\phi}\)
simultaneously with showing for games $\alpha$ with desired instance $\tilde{\alpha}$ that there is a uniform substitution $\sigma$ such that \(\usubstappp{\allvars}{\allvars}{\sigma}{\alpha} = \tilde{\alpha}\).
The output taboo $W$ of \(\usubstappp{W}{\allvars}{\sigma}{\alpha}\) equals $\allvars$ by \rref{lem:usubst-taboos}, because all variables $\allvars$ are already input taboos.
Nothing needs to be shown for terms as game symbols cannot occur in terms.
\begin{compactenum}
\item Case \(\phi\land\psi\) with desired instance \(\tilde{\phi}\land\tilde{\psi}\)
  (which has to have this shape to qualify as a schema instance).
  By IH, there are substitutions $\sigma,\tau$ such that \(\usubstapp{\allvars}{\sigma}{\phi}=\tilde{\phi}\)
  and \(\usubstapp{\allvars}{\tau}{\psi}=\tilde{\psi}\).
  The union \(\usubstjoin{\phi}{\psi}\) is defined, because the same replacements have been used consistently in all occurrences of the instantiation.
  Thus, \(\usubstapp{\allvars}{(\usubstjoin{\sigma}{\tau})}{\usubstgroup{\phi\land\psi}}
  =\usubstapp{\allvars}{(\usubstjoin{\sigma}{\tau})}{\phi} \land \usubstapp{\allvars}{(\usubstjoin{\sigma}{\tau})}{\psi}
  = \usubstapp{\allvars}{\sigma}{\phi} \land \usubstapp{\allvars}{\tau}{\psi} 
  = \tilde{\phi} \land \tilde{\psi}\)
  as desired.
  The proof is accordingly for $\lnot$ etc.
  
\item Case \(\lexists{x}{\phi}\) with desired instance \(\lexists{x}{\tilde{\phi}}\).
  By IH, there is a substitution $\sigma$ such that \(\usubstapp{\allvars}{\sigma}{\phi}=\tilde{\phi}\).
  Thus, \(\usubstapp{\allvars}{\sigma}{\usubstgroup{\lexists{x}{\phi}}}
  = \lexists{x}{\usubstapp{\allvars\cup\{x\}}{\sigma}{\phi}}
  = \lexists{x}{\usubstapp{\allvars}{\sigma}{\phi}}
  = \lexists{x}{\tilde{\phi}}\)
  as desired.

\item Case \(\ddiamond{\alpha}{\phi}\) with desired instance \(\ddiamond{\tilde{\alpha}}{\tilde{\phi}}\).
  By IH, there are substitutions $\sigma,\tau$ such that \(\usubstappp{\allvars}{\allvars}{\sigma}{\alpha}=\tilde{\alpha}\)
  and \(\usubstapp{\allvars}{\tau}{\phi}=\tilde{\phi}\).
  Thus, the union \(\usubstjoin{\sigma}{\tau}\) is defined and
  \(\usubstapp{\allvars}{(\usubstjoin{\sigma}{\tau})}{\usubstgroup{\ddiamond{\alpha}{\psi}}}
  =\ddiamond{\usubstappp{\allvars}{\allvars}{(\usubstjoin{\sigma}{\tau})}{\alpha}} {\usubstapp{\allvars}{(\usubstjoin{\sigma}{\tau})}{\phi}}
  = \ddiamond{\usubstappp{\allvars}{\allvars}{\sigma}{\alpha}} {\usubstapp{\allvars}{\tau}{\phi}} 
  = \ddiamond{\tilde{\alpha}}{\tilde{\phi}}\)
  as desired.

\item Case $a$ of a game symbol with desired instance $\tilde{\alpha}$
  is handled with the substitution \(\sigma = \usubstlist{\usubstmod{a}{\tilde{\alpha}}}\),
  which satisfies \(\usubstappp{\allvars}{\allvars}{\sigma}{a}=\applyusubst{\sigma}{a}=\tilde{\alpha}\)
  as desired.
  
\item Case \(\pevolvein{\D{x}=\theta}{\psi}\) with desired instance \(\pevolvein{\D{x}=\tilde{\theta}}{\tilde{\psi}}\).
  By IH, there are substitutions $\sigma,\tau$ such that \(\usubstapp{\allvars}{\sigma}{\theta}=\tilde{\theta}\)
  and \(\usubstapp{\allvars}{\tau}{\psi}=\tilde{\psi}\).
  Thus, the union \(\usubstjoin{\sigma}{\tau}\) is defined and
  \(\usubstapp{\allvars}{(\usubstjoin{\sigma}{\tau})}{\theta}
  = \usubstapp{\allvars}{\sigma}{\theta}=\tilde{\theta}\)
  and \(\usubstapp{\allvars}{(\usubstjoin{\sigma}{\tau})}{\psi}
  = \usubstapp{\allvars}{\tau}{\psi}=\tilde{\psi}\),
  hence,
  \(\usubstapp{\allvars}{(\usubstjoin{\sigma}{\tau})}{\usubstgroup{\pevolvein{\D{x}=\theta}{\psi}}}
  = (\pevolvein{\D{x}=\tilde{\theta}}{\tilde{\psi}})\)
  as desired.
  
\item Case \(\pchoice{\alpha}{\beta}\) with desired instance \(\pchoice{\tilde{\alpha}}{\tilde{\beta}}\).
  By IH there are substitutions $\sigma,\tau$ such that \(\usubstappp{\allvars}{\allvars}{\sigma}{\alpha}=\tilde{\alpha}\)
  and \(\usubstappp{\allvars}{\allvars}{\tau}{\beta}=\tilde{\beta}\).
  Thus, the union \(\usubstjoin{\sigma}{\tau}\) is defined and
  \(\usubstapp{\allvars}{(\usubstjoin{\sigma}{\tau})}{\usubstgroup{\pchoice{\alpha}{\beta}}}
  =\pchoice{\usubstappp{\allvars}{\allvars}{(\usubstjoin{\sigma}{\tau})}{\alpha}} {\usubstappp{\allvars}{\allvars}{(\usubstjoin{\sigma}{\tau})}{\beta}}
  = \pchoice{\usubstappp{\allvars}{\allvars}{\sigma}{\alpha}}{\usubstappp{\allvars}{\allvars}{\tau}{\beta}}
  = \pchoice{\tilde{\alpha}}{\tilde{\beta}}\)
  using that \(\allvars=\allvars\cup\allvars\).

\item Case \(\alpha;\beta\) with desired instance \(\tilde{\alpha};\tilde{\beta}\).
  By IH there are substitutions $\sigma,\tau$ such that \(\usubstappp{\allvars}{\allvars}{\sigma}{\alpha}=\tilde{\alpha}\)
  and \(\usubstappp{\allvars}{\allvars}{\tau}{\beta}=\tilde{\beta}\).
  Thus, the union \(\usubstjoin{\sigma}{\tau}\) is defined and
  \(\usubstapp{\allvars}{(\usubstjoin{\sigma}{\tau})}{\usubstgroup{\alpha;\beta}}
  =\usubstappp{\allvars}{\allvars}{(\usubstjoin{\sigma}{\tau})}{\alpha}; \usubstappp{\allvars}{\allvars}{(\usubstjoin{\sigma}{\tau})}{\beta}
  = \usubstappp{\allvars}{\allvars}{\sigma}{\alpha}; \usubstappp{\allvars}{\allvars}{\tau}{\beta} 
  = \tilde{\alpha};\tilde{\beta}\)
  as desired.

\item Case \(\prepeat{\alpha}\) with desired instance \(\prepeat{\tilde{\alpha}}\).
  By IH there is a substitution $\sigma$ such that \(\usubstappp{\allvars}{\allvars}{\sigma}{\alpha}=\tilde{\alpha}\).
  Thus,
  \(\usubstapp{\allvars}{\sigma}{\usubstgroup{\prepeat{\alpha}}}
  =\prepeat{(\usubstappp{\allvars}{\allvars}{\sigma}{\alpha})}
  = \prepeat{\tilde{\alpha}}\)
  as desired, because \(\usubstappp{\allvars}{\allvars}{\sigma}{\alpha}\) is defined.
\end{compactenum}
Case $\pdual{\alpha}$ is accordingly.
Axiomatic proof rules built from surjective formulas are surjective, because \irref{USR} can instantiate the rule to any instance as long as \irref{US} can instantiate all premises and the conclusion to any instance.
\qedhere
\end{proofatend}

Instead of following previous completeness arguments for uniform substitution \cite{DBLP:conf/cade/Platzer18}, this paper presents a pure game-style uniform substitution formulation in \rref{fig:dGL} of a \dGL axiomatization that makes the overall completeness proof most straightforward.
For that purpose, the \dGL axiomatization in \rref{fig:dGL} uses properties $\pusall$ of a game symbol $c$, which, as a game, can impose arbitrary conditions on the state even for a trivial postcondition (the formula $\ltrue$ is always true).

\begin{figure}[tb]
  \centering
  \renewcommand*{\irrulename}[1]{\text{#1}}%
  \renewcommand{\linferenceRuleNameSeparation}{~~}%
  \newdimen\linferenceRulehskipamount%
  \linferenceRulehskipamount=1mm%
  \linferenceRulevskipamount=0.4em%
  \begin{calculuscollections}{\columnwidth}
    \begin{calculus}
      \cinferenceRule[box|$\dibox{\cdot}$]{box axiom}
      {\linferenceRule[equiv]
        {\lnot\ddiamond{a}{\lnot \pusall}}
        {\axkey{\dbox{a}{\pusall}}}
      }
      {}
      \cinferenceRule[assigndeq|$\didia{:=}_{=}$]{assignment / equational axiom}
      {\linferenceRule[equiv]
        {\lexists{x}{(x=f \land \pusall)}}
        {\axkey{\ddiamond{\pupdate{\umod{x}{f}}}{\pusall}}}
      }
      {}%
      \cinferenceRule[DSd|DS]{(constant) differential equation solution} %
      {\linferenceRule[viuqe]
        {\axkey{\ddiamond{\pevolve{\D{x}=f}}{\pusall}}}
        {\lexists{t{\geq}0}{\ddiamond{\pupdate{\pumod{x}{x{+}f\itimes t}}}{\ddiamond{\pupdate{\pumod{\D{x}}{f}}}{\pusall}}}}
      }
      {}
      \cinferenceRule[testd|$\didia{?}$]{test}
      {\linferenceRule[equiv]
        {q \land p}
        {\axkey{\ddiamond{\ptest{q}}{p}}}
      }{}
      \cinferenceRule[choiced|$\didia{\cup}$]{axiom of nondeterministic choice}
      {\linferenceRule[equiv]
        {\ddiamond{a}{\pusall} \lor \ddiamond{b}{\pusall}}
        {\axkey{\ddiamond{\pchoice{a}{b}}{\pusall}}}
      }{}
      \cinferenceRule[composed|$\didia{{;}}$]{composition}
      {\linferenceRule[equiv]
        {\ddiamond{a}{\ddiamond{b}{\pusall}}}
        {\axkey{\ddiamond{a;b}{\pusall}}}
      }{}
      \cinferenceRule[iterated|$\didia{{}^*}$]{iteration/repeat unwind pre-fixpoint, even fixpoint}
      {\linferenceRule[equiv]
        {\pusall \lor \ddiamond{a}{\ddiamond{\prepeat{a}}{\pusall}}}
        {\axkey{\ddiamond{\prepeat{a}}{\pusall}}}
      }{}
      \cinferenceRule[duald|$\didia{{^d}}$]{dual}
      {\linferenceRule[equiv]
        {\lnot\ddiamond{a}{\lnot \pusall}}
        {\axkey{\ddiamond{\pdual{a}}{\pusall}}}
      }{}
    \end{calculus}%
    \hfill%
    \begin{calculus}
      \cinferenceRule[M|M]{$\ddiamond{}{}$ monotone / $\ddiamond{}{}$-generalization} %
      {\linferenceRule[formula]
        {\pusall\limply \qusall}
        {\ddiamond{a}{\pusall}\limply\ddiamond{a}{\qusall}}
      }{}
      \cinferenceRule[FP|FP]{iteration is least fixpoint / reflexive transitive closure RTC, equivalent to invind in the presence of R}
      {\linferenceRule[formula]
        {\pusall \lor \ddiamond{a}{\qusall} \limply \qusall}
        {\ddiamond{\prepeat{a}}{\pusall} \limply \qusall}
      }{}
      \cinferenceRule[MP|MP]{modus ponens}
      {\linferenceRule[formula]
        {p \quad p\limply q}
        {q}
      }{}
      \cinferenceRule[gena|$\forall$]{$\forall{}$ generalisation}
      {\linferenceRule[formula]
        {\pusall}
        {\lforall{x}{\pusall}}
      }{}
    \end{calculus}%
  \end{calculuscollections}
  \vspace*{-\baselineskip} %
  \caption{Differential game logic axioms and axiomatic proof rules}
  \label{fig:dGL}
\end{figure}%

All axioms of \rref{fig:dGL}, except test \irref{testd}, equational assignment \irref{assigndeq}, and constant solution \irref{DSd}, are surjective by \rref{lem:surjectiveaxiom}.
The \irref{US} requirement that no substitute of $f$ may depend on $x$ is important for the soundness of \irref{DSd} and \irref{assigndeq}.
Axiom \irref{testd} is surjective, as it has no bound variables, so generates no taboos and none of its instances clash:
\(
\usubstapp{\emptyset}{\sigma}{\usubstgroup{\ddiamond{\ptest{q}}{p} \lbisubjunct q \land p}}
=
(\ddiamond{\usubstappp{\emptyset}{\emptyset}{\sigma}{q}}{\usubstapp{\emptyset}{\sigma}{p}} \lbisubjunct \usubstapp{\emptyset}{\sigma}{q} \land \usubstapp{\emptyset}{\sigma}{p})
\).
Similarly, rule \irref{MP} is surjective \cite{DBLP:journals/jar/Platzer17}, and the other rules are surjective by \rref{lem:surjectiveaxiom}.
Other differential equation axioms are elided but work as previously \cite{DBLP:journals/jar/Platzer17}.

Besides rule \irref{US}, \emph{bound variable renaming} (rule \irref{BR}) is the only schematic principle, mostly for generalizing assignment axiom \irref{assigndeq} to other variables.

\begin{lemma}[Bound renaming] \label{lem:brename}
Rule \irref{BR} is locally sound, where $\urename[\psi]{x}{y}$ is the result of uniformly renaming $x$ to $y$ in $\psi$ (also $\D{x}$ to $\D{y}$ but no $\D[2]{x},\D[3]{x}$ etc.\ or game symbols occur in $\psi$, where the rule \irref{BR} for \(\dbox{\pupdate{\pumod{x}{\theta}}}{\psi}\) is accordingly):%
{\upshape\[
\cinferenceRule[BR|BR]{bound renaming}
{\linferenceRule[formula]
  {\phi \limply \ddiamond{\pupdate{\pumod{y}{\theta}}}{\ddiamond{\Dupdate{\Dumod{\D{y}}{\D{x}}}}{\urename[\psi]{x}{y}}}}
  {\phi \limply \ddiamond{\pupdate{\pumod{x}{\theta}}}{\psi}}
}{y,\D{y}\not\in\psi}
\]}%
\end{lemma}
\begin{proofatend}
This proof is the only one using that no higher-order differential variables $\D[(i)]{x}$ for $i\geq2$ occur.
It also assumes that no game symbols $a$ occur, because \(\urename[a]{x}{y}\) has no syntactic representation.
Local soundness follows from:
\[
\ddiamond{\pupdate{\pumod{x}{\theta}}}{\psi}
\lbisubjunct \ddiamond{\pupdate{\pumod{y}{\theta}}}{\ddiamond{\Dupdate{\Dumod{\D{y}}{\D{x}}}}{\urename[\psi]{x}{y}}}
\quad(y,\D{y}\not\in\psi)
\]
Consider any state $\iget[state]{\I}$ in which to show this equivalence.
Then \(\imodels{\I}{\ddiamond{\pupdate{\pumod{x}{\theta}}}{\psi}}\)
iff \(\modif{\iget[state]{\I}}{x}{\ivaluation{\I}{\theta}} \in \imodel{\I}{\psi}\)
iff, by \rref{eq:urenameadjoint} below, 
\({\modif{\iget[state]{\I}}{y}{\ivaluation{\I}{\theta}}}\modif{\,}{\D{y}}{\iget[state]{\I}(\D{x})} \in \imodel{\I}{\urename[\psi]{x}{y}}\)
iff \(\imodels{\I}{\ddiamond{\pupdate{\pumod{y}{\theta}}}{\ddiamond{\Dupdate{\Dumod{\D{y}}{\D{x}}}}{\urename[\psi]{x}{y}}}}\).
The values of $x,\D{x}$ are irrelevant for $\urename[\psi]{x}{y}$ by \rref{lem:coincidence}. No $\D[(i)]{y}$ for $i\geq2$ occur.
It uses a fact about uniform renaming of $\D[(i)]{x}$ to $\D[(i)]{y}$ and vice versa, for all $i$:
\begin{equation}
\imodels{\I}{\psi} ~~\text{iff}~~ {\modif{\iget[state]{\I}}{\D[(i)]{x}}{\iget[state]{\I}(\D[(i)]{y})}}\modif{}{\D[(i)]{y}}{\iget[state]{\I}(\D[(i)]{x})} \in \imodel{\I}{\urename[\psi]{x}{y}}
~\text{where the state is modified for all~$i$}
\tag{$\ast$}
\label{eq:urenameadjoint}
\end{equation}
Property \rref{eq:urenameadjoint} is proved by straightforward induction on the structure of $\psi$ using that $x$ and $\D{x}$ etc.\ are consistently swapped with $y$ and $\D{y}$ etc.\ syntactically in the uniformly renamed formula \(\urename[\psi]{x}{y}\) as well as semantically in the state.
\qedhere
\end{proofatend}

\newcommand{\reduct}[1]{#1^\flat}%
\newcommand{\LBase}{\textit{L}\xspace}%

\begin{theorem}[Relative completeness] \label{thm:dGL-complete}%
  The \dGL calculus is a \emph{sound and complete axiomatization} of hybrid games relative to \emph{any} differentially expressive logic \LBase, i.e., every valid \dGL formula is provable in \dGL from \LBase tautologies.
\end{theorem}
\begin{proofatend}
\let\Oracle\LBase%
\newcommand{\precondf}{F}%
\newcommand{\postcondf}{G}%
The axioms and axiomatic rules in \rref{fig:dGL} are concrete instances of sound schemata or rules from prior work \cite{DBLP:journals/tocl/Platzer15,DBLP:journals/jar/Platzer17} except for a slight modification in axiom \irref{DSd}, which is sound, because the effect of a differential equation \(\pevolve{\D{x}=f}\) on $\D{x}$ is that its value equals $f$ while following the ODE.

The completeness proof is by induction on a well-founded partial order~$\prec$ induced by the lexicographic ordering of the overall structural complexity of the hybrid games in the formula and the structural complexity of the formula itself, with the logic \LBase placed at the bottom of the partial order \cite{DBLP:journals/tocl/Platzer15}.
Even if all axioms and rules in \rref{fig:dGL} except \irref{assigndeq+DSd} are surjective by \rref{lem:surjectiveaxiom}, most do not have the form used in the schematic completeness result for \dGL \cite[Thm.\,4.5]{DBLP:journals/tocl/Platzer15}.
All required schematic instances of all axioms (except assignments) for that completeness result can, nevertheless, be obtained by instantiating game symbol $c$ to the test game $\ptest{\psi}$ for the desired instance $\psi$, which is possible by \rref{lem:surjectiveaxiom}.
Uniform substitution then turns each respective occurrence of $\pusall$ into \(\ddiamond{\ptest{\psi}}{\ltrue}\),
which an additional use of surjective axiom \irref{testd} turns into \(\psi\land\ltrue\), which first-order logic equivalences in \LBase simplify to the desired $\psi$.

  For example, consider the representative case \m{\entails \precondf \limply \ddiamond{\pdual{\beta}}{\postcondf}}, which implies
   \m{\entails \precondf \limply \lnot\ddiamond{\beta}{\lnot \postcondf}}, which implies
   \m{\entails \precondf \limply \dbox{\beta}{\postcondf}}.
   Since \(\dbox{\beta}{\postcondf} \prec \ddiamond{\pdual{\beta}}{\postcondf}\), because $\pdual{\beta}$ is more complex than $\beta$ even if the modality changed,
   \m{\infers[\Oracle] \precondf \limply \dbox{\beta}{\postcondf}} can be derived by IH.
   Axiom \irref{box}, thus, derives
   \m{\infers[\Oracle] \precondf \limply \lnot\ddiamond{\beta}{\lnot \postcondf}},
   from which, with \rref{lem:surjectiveaxiom} and the above observations about axiom \irref{testd}, axiom \irref{duald} derives 
   \m{\infers[\Oracle] \precondf \limply \ddiamond{\pdual{\beta}}{\postcondf}}.
   
Thus, \rref{lem:surjectiveaxiom} makes the previous completeness proof \cite[Thm.\,4.5]{DBLP:journals/tocl/Platzer15} with the uniform substitution relative completeness refinements \cite[Thm.\,40]{DBLP:journals/jar/Platzer17} transfer to \rref{fig:dGL},
but only if all uses of the assignment axiom, which is not surjective, can be patched.
The only such case is in the proof that
\m{\entails \precondf \limply \ddiamond{\pupdate{\pumod{x}{\theta}}}{\postcondf}}
implies that this formula can be proved in the \dGL calculus from \Oracle, which, because of the different axioms, works differently than the corresponding case of \m{\entails \precondf \limply \dbox{\pupdate{\pumod{x}{\theta}}}{\postcondf}} in the completeness proof for \dL \cite[Thm.\,40]{DBLP:journals/jar/Platzer17}.

   If \m{\entails \precondf \limply \ddiamond{\pupdate{\pumod{y}{\theta}}}{\postcondf}}, then this formula can be proved, using a fresh variable $x$ not occurring in $\theta$ or $\postcondf$, with the following derivation by renaming (\rref{lem:brename})
  \begin{sequentdeduction}[array]
    \linfer[BR] %
    {\linfer[assigndeq] %
      {\linfer[assigndeq] %
        {\lsequent{\precondf} {\lexists{x}{(x=\theta\land\lexists{\D{x}}{(\D{x}=\D{y}\land\urename[\postcondf]{y}{x})})}}}
      {\lsequent{\precondf} {\lexists{x}{(x=\theta\land\urename[\ddiamond{\Dupdate{\Dumod{\D{x}}{\D{y}}}}{\postcondf}]{y}{x})}}}
      }
    {\lsequent{\precondf} {\ddiamond{\pupdate{\pumod{x}{\theta}}}{\ddiamond{\Dupdate{\Dumod{\D{x}}{\D{y}}}}{\urename[\postcondf]{y}{x}}}}}
    }%
    {\lsequent{\precondf} {\ddiamond{\pupdate{\pumod{y}{\theta}}}{\postcondf}}}
  \end{sequentdeduction}
  In the above proof, the two instantiations of axiom \irref{assigndeq} succeed, because $x$ and $\D{x}$ are fresh, so do not occur in either $\theta$ or $\D{y}$.
  The above proof only used equivalence transformations, so its premise is valid iff its conclusion is, which it is by assumption, so implies 
  \(\entails \precondf \limply \lexists{x}{(x=\theta\land\lexists{\D{x}}{(\D{x}=\D{y}\land\urename[\postcondf]{y}{x})})}\).
  Since \(\big(\precondf \limply \lexists{x}{(x=\theta\land\lexists{\D{x}}{(\D{x}=\D{y}\land\urename[\postcondf]{y}{x})})}\big) \prec (\precondf \limply \ddiamond{\pupdate{\pumod{y}{\theta}}}{\postcondf})\),
  because there are less hybrid games,
  \(\infers[\Oracle] \precondf \limply \lexists{x}{(x=\theta\land\lexists{\D{x}}{(\D{x}=\D{y}\land\urename[\postcondf]{y}{x})})}\)
  by IH.
  The above proof, thus, derives \(\infers[\Oracle] \precondf \limply \ddiamond{\pupdate{\pumod{y}{\theta}}}{\postcondf}\).
  For later, also note the derivability of:
  \begin{equation}
  \postcondf \lbisubjunct \ddiamond{\pupdate{\pumod{x}{x}}}{\postcondf}
  \label{eq:stutterd}
  \end{equation}
  Since it is valid, this stuttering identity derives with an additional derivation of the converse \(\ddiamond{\pupdate{\pumod{x}{x}}}{\postcondf} \limply \postcondf\).
  That follows from similarly deriving \(\ddiamond{\pupdate{\pumod{x}{x}}}{\postcondf} \limply \precondf\) by contraposition like above with a fresh $x$ if \(\entails \ddiamond{\pupdate{\pumod{x}{x}}}{\postcondf} \limply \precondf\):
  \begin{sequentdeduction}[array]
  \linfer%
  {\linfer[box]
    {\linfer[BR] %
    {\linfer[box]
    {\linfer[assigndeq] %
      {\linfer[assigndeq] %
        {\lsequent{\lnot\precondf} {\lnot\lexists{x}{(x=\theta\land\lexists{\D{x}}{(\D{x}=\D{y}\land\urename[\postcondf]{y}{x})})}}}
      {\lsequent{\lnot\precondf} {\lnot\lexists{x}{(x=\theta\land\urename[\ddiamond{\Dupdate{\Dumod{\D{x}}{\D{y}}}}{\postcondf}]{y}{x})}}}
      }
    {\lsequent{\lnot\precondf} {\lnot\ddiamond{\pupdate{\pumod{x}{\theta}}}{\ddiamond{\Dupdate{\Dumod{\D{x}}{\D{y}}}}{\urename[\postcondf]{y}{x}}}}}
    }%
    {\lsequent{\lnot\precondf} {\dbox{\pupdate{\pumod{x}{\theta}}}{\dbox{\Dupdate{\Dumod{\D{x}}{\D{y}}}}{\urename[\lnot\postcondf]{y}{x}}}}}
    }%
    {\lsequent{\lnot\precondf} {\dbox{\pupdate{\pumod{y}{\theta}}}{\lnot\postcondf}}}
    }%
    {\lsequent{\lnot\precondf} {\lnot\ddiamond{\pupdate{\pumod{y}{\theta}}}{\postcondf}}}
    }%
  {\lsequent{\ddiamond{\pupdate{\pumod{y}{\theta}}}{\postcondf}} {\precondf}}
  \end{sequentdeduction}

A final subtlety arises in the case of diamond properties of loops \cite{DBLP:journals/jar/Platzer17}. Let
\m{\entails \precondf \limply \ddiamond{\prepeat{\beta}}{\postcondf}}.
   \def\vec#1{#1}%
    Let $\vec{x}$ be the (\emph{finite!}) vector of free variables \m{\freevars{\ddiamond{\prepeat{\beta}}{\postcondf}}}.
    Since $\ddiamond{\prepeat{\beta}}{\postcondf}$ is a least pre-fixpoint \cite{DBLP:journals/tocl/Platzer15}, for all \dGL formulas $\psi$ with \(\freevars{\psi}\subseteq\freevars{\ddiamond{\prepeat{\beta}}{\postcondf}}\):
    \[
    \entails \lforall{\vec{x}}{(\postcondf\lor\ddiamond{\beta}{\psi}\limply\psi)} \limply (\ddiamond{\prepeat{\beta}}{\postcondf} \limply\psi)
    \]
    In particular, this holds for a fresh predicate symbol $p$ with arguments $\vec{x}$:
    \[
    \entails \lforall{\vec{x}}{(\postcondf\lor\ddiamond{\beta}{p(\vec{x})}\limply p(\vec{x}))} \limply (\ddiamond{\prepeat{\beta}}{\postcondf} \limply p(\vec{x}))
    \]
    Using \m{\entails \precondf \limply \ddiamond{\prepeat{\beta}}{\postcondf}}, this implies
    \[
    \entails \lforall{\vec{x}}{(\postcondf\lor\ddiamond{\beta}{p(\vec{x})}\limply p(\vec{x}))} \limply (\precondf \limply p(\vec{x}))
    \]
    As \((\lforall{\vec{x}}{(\postcondf\lor\ddiamond{\beta}{p(\vec{x})}\limply p(\vec{x}))} \limply (\precondf \limply p(\vec{x}))) \prec \phi\), because, even if the formula complexity increased, the structural complexity of the games decreased, since $\phi$ has one more repetition, this fact is derivable by IH:
    \[
    \infers[\Oracle] \lforall{\vec{x}}{(\postcondf\lor\ddiamond{\beta}{p(\vec{x})}\limply p(\vec{x}))} \limply (\precondf \limply p(\vec{x}))
    \]
    The uniform substitution \(\sigma=\usubstlist{\usubstmod{p(\usarg)}{\ddiamond{\pupdate{\pumod{\vec{x}}{\usarg}}}{\ddiamond{\prepeat{\beta}}{\postcondf}}}}\) does not clash since \(\freevars{\ddiamond{\prepeat{\beta}}{\postcondf}}\subseteq\{\vec{x}\}\).
    Since $p$ does not occur in $\precondf$, $\postcondf$ or $\beta$, rule \irref{US} derives:
    \begin{sequentdeduction}[array]
    \linfer%
    {\linfer[US]
      {\lsequent{} {\lforall{\vec{x}}{(\postcondf\lor\ddiamond{\beta}{p(\vec{x})}\limply p(\vec{x}))} \limply (\precondf \limply p(\vec{x}))}}
      {\lsequent{} {\lforall{\vec{x}}{(\postcondf\lor\ddiamond{\beta}{\ddiamond{\pupdate{\pumod{\vec{x}}{\vec{x}}}}{\ddiamond{\prepeat{\beta}}{\postcondf}}}\limply \ddiamond{\pupdate{\pumod{\vec{x}}{\vec{x}}}}{\ddiamond{\prepeat{\beta}}{\postcondf}})} \limply (\precondf \limply \ddiamond{\pupdate{\pumod{\vec{x}}{\vec{x}}}}{\ddiamond{\prepeat{\beta}}{\postcondf}})}}
    }%
    {\lsequent{} {\lforall{\vec{x}}{(\postcondf\lor\ddiamond{\beta}{\ddiamond{\prepeat{\beta}}{\postcondf}}\limply \ddiamond{\prepeat{\beta}}{\postcondf})} \limply (\precondf \limply \ddiamond{\prepeat{\beta}}{\postcondf})}}
    \end{sequentdeduction}
    where the last inference used the derivable stuttering identity \rref{eq:stutterd} three times.
    The iteration axiom \irref{iterated} with \rref{lem:surjectiveaxiom} completes this derivation:
\renewcommand{\linferPremissSeparation}{~}%

\begin{minipage}{\textwidth}
\advance\leftskip-0.6cm
\begin{minipage}{\textwidth}%
\renewcommand{\linferPremissSeparation}{}%
    \begin{sequentdeduction}[Hilbert+array]
    \linfer[MP]
    {
    \lsequent{} {\lforall{\vec{x}}{(\postcondf\lor\ddiamond{\beta}{\ddiamond{\prepeat{\beta}}{\postcondf}}\limply \ddiamond{\prepeat{\beta}}{\postcondf})} \limply (\precondf \limply \ddiamond{\prepeat{\beta}}{\postcondf})}
    !
    \linfer[gena]
    {\linfer[iterated]
      {\lclose}
      {\lsequent{} {\postcondf\lor\ddiamond{\beta}{\ddiamond{\prepeat{\beta}}{\postcondf}} \limply \ddiamond{\prepeat{\beta}}{\postcondf}}}
    }
    {\lsequent{} {\lforall{\vec{x}}{(\postcondf\lor\ddiamond{\beta}{\ddiamond{\prepeat{\beta}}{\postcondf}} \limply \ddiamond{\prepeat{\beta}}{\postcondf})}}}
    }
    {\lsequent{\precondf} {\ddiamond{\prepeat{\beta}}{\postcondf}}}
    \end{sequentdeduction}
    \end{minipage}
    \end{minipage}
    Observe that rules \irref{gena} and \irref{MP} instantiate as needed with \irref{USR} by \rref{lem:surjectiveaxiom}.
\qedhere
\end{proofatend}

This completeness result assumes that no game symbols occur, because uniform renaming otherwise needs to become a syntactic operator.
{%
A logic \LBase closed under first-order connectives
is \emph{differentially expressive} (for \dGL) if every \dGL formula $\phi$ has an equivalent $\reduct{\phi}$ in \LBase and all differential equation equivalences of the form \(\ddiamond{\pevolve{\D{x}=\genDE{x}}}{G} \lbisubjunct \reduct{(\ddiamond{\pevolve{\D{x}=\genDE{x}}}{G})}\) for $G$ in \LBase are provable in its calculus.
}

\section{Differential Hybrid Games}

\providecommand*{\pdiffgame}[3]{{#1}{\ifthenelse{\equal{#2}{}}{}{{}\&^{\hspace{-3pt}d}#2}\ifthenelse{\equal{#3}{}}{}{\&}#3}}%
\newcommand*{\genDG}[3]{\theta}%
\newcommand*{\genDGb}[3]{\eta}%
\newcommand*{\pdifftabooaugment}[4]{\textcolor{blue}{\bar{#1}}}%
\newcommand*{\pdifftabooaugmentdef}[4]{#1\cup\{#2,\D{#2},#3,\D{#3},#4,\D{#4}\}}
\newcommand*{\urengroup}[1]{(#1)}%

Uniform substitution generalizes from \dGL for hybrid games \cite{DBLP:journals/tocl/Platzer15} to \dGL for \emph{differential} hybrid games \cite{DBLP:journals/tocl/Platzer17}, which add differential games as a new atomic game.
A \emph{differential game} of the form \(\pdiffgame{\D{x}=\genDG{x}{y}{z}}{y\in Y}{z\in Z}\)
allows Angel to control how long to follow the differential equation \(\D{x}=\genDG{x}{y}{z}\) (in which variables $x,y,z$ may occur) while Demon provides a measurable input for $y$ over time satisfying the formula $y\in Y$ always and Angel, knowing Demon's current input, provides a measurable input for $z$ satisfying the formula $z\in Z$. 
All occurrences of $y,z$ in \(\pdiffgame{\D{x}=\genDG{x}{y}{z}}{y\in Y}{z\in Z}\) are bound, and $y\in Y$ and $z\in Z$ are formulas in the free variables $y$ or $z$, respectively.
It has been a long-standing challenge to give mathematical meaning \cite{DBLP:journals/tams/ElliottK74,DBLP:journals/indianam/EvansSouganidis84} and sound reasoning principles
\cite{DBLP:journals/tocl/Platzer17} for differential games. Both outcomes can simply be adopted here under the usual well-definedness assumptions \cite{DBLP:journals/tocl/Platzer17}.

Uniform substitution application in \rref{fig:usubst-one} lifts to differential games by adding:
\[
    \usubstappp{\pdifftabooaugment{U}{x}{y}{z}}{U}{\sigma}{\usubstgroup{\pdiffgame{\D{x}=\genDG{x}{y}{z}}{y\in Y}{z\in Z}}}
    \mnodefeq
    (\pdiffgame{\D{x}=\usubstapp{\pdifftabooaugment{U}{x}{y}{z}}{\sigma}{\genDG{x}{y}{z}}}{y\in\usubstapp{\pdifftabooaugment{U}{x}{y}{z}}{\sigma}{Y}}{z\in\usubstapp{\pdifftabooaugment{U}{x}{y}{z}}{\sigma}{Z}})
\]
where \(\pdifftabooaugment{U}{x}{y}{z}\) is \(\pdifftabooaugmentdef{U}{x}{y}{z}\).
Well-definedness assumptions on differential games \cite{DBLP:journals/tocl/Platzer17} need to hold, e.g., only first-order logic formulas denoting compact sets are allowed for controls and the differential equations need to be bounded.

As terms are unaffected by adding differential games to the syntax, \rref{lem:coincidence-term} and~\ref{lem:usubst-term} do not change.
The proofs of the coincidence lemmas~\ref{lem:coincidence} and~\ref{lem:coincidence-HG} and bound effect lemma~\ref{lem:bound} \cite{DBLP:conf/cade/Platzer18} transfer to \dGL with differential hybrid games in verbatim thanks to their use of \emph{semantically defined} free and bound variables, which carry over to differential hybrid games.
The proof of \rref{lem:usubst-taboos} generalizes easily by adding a case for differential games with the above \(\pdifftabooaugment{U}{x}{y}{z}\).
The uniform substitution lemmas~\ref{lem:usubst} and \ref{lem:usubst-HG} inductively generalize to differential hybrid games because of:%
\begin{lemma}[Uniform substitution for differential games]
Let \(U\subseteq\allvars\).
For all $U$-variations $\iget[state]{\Ie}$ of $\iget[state]{\I}$:
\[\iwin[{\usubstappp{\pdifftabooaugment{U}{x}{y}{z}}{U}{\sigma}{\usubstgroup{\pdiffgame{\D{x}=\genDG{x}{y}{z}}{y\in Y}{z\in Z}}}}]{\Ie}{X} ~\text{iff}~ \iwin[\pdiffgame{\D{x}=\genDG{x}{y}{z}}{y\in Y}{z\in Z}]{\Iae}{X}\]
\end{lemma}
\begin{proofatend}
\newcommand{\Iazyv}{\vdLint[const=\sigma^*_\omega I,state={\urename[\mu]{y}{v}}]}%
\newcommand{\Iazzw}{\vdLint[const=\sigma^*_\omega I,state={\urename[\mu]{z}{w}}]}%
\newcommand{\Izyvzw}{\vdLint[const=I,state={\urename[\mu]{y}{v}\urename{z}{w}}]}%
\newcommand{\Iazyvzw}{\vdLint[const=\sigma^*_\omega I,state={\urename[\mu]{y}{v}\urename{z}{w}}]}%
Left side is
\(\iwin[{\usubstappp{\pdifftabooaugment{U}{x}{y}{z}}{U}{\sigma}{\usubstgroup{\pdiffgame{\D{x}=\genDG{x}{y}{z}}{y\in Y}{z\in Z}}}}]{\Ie}{X}\)
= \(\iwinreg[\pdiffgame{\D{x}{=}\usubstapp{\pdifftabooaugment{U}{x}{y}{z}}{\sigma}{\genDG{x}{y}{z}}}{y\in\usubstapp{\pdifftabooaugment{U}{x}{y}{z}}{\sigma}{Y}}{z\in\usubstapp{\pdifftabooaugment{U}{x}{y}{z}}{\sigma}{Z}}]{\Ie}{X}\)
= \(\iwinreg[\pdiffgame{\D{x}{=}\urename[\usubstapp{\pdifftabooaugment{U}{x}{y}{z}}{\sigma}{\genDG{x}{y}{z}}]{y}{v}\urename{z}{w}}{v\in\usubstapp{\pdifftabooaugment{U}{x}{y}{z}}{\sigma}{Y}}{w\in\usubstapp{\pdifftabooaugment{U}{x}{y}{z}}{\sigma}{Z}}]{\Ie}{X}\)
by uniform renaming of $y$ to $v$ and $z$ to $w$ (proof of \rref{lem:brename}), which are fresh.
Here \(\urename[\usubstapp{\pdifftabooaugment{U}{x}{y}{z}}{\sigma}{\genDG{x}{y}{z}}]{y}{v}\urename{z}{w}\) is the result of uniformly renaming $y$ to $v$ and $z$ to $w$ in the term \(\usubstapp{\pdifftabooaugment{U}{x}{y}{z}}{\sigma}{\genDG{x}{y}{z}}\)
and
\(v\in\usubstapp{\pdifftabooaugment{U}{x}{y}{z}}{\sigma}{Y}\) the result of uniformly renaming $y$ to $v$ in \(y\in\usubstapp{\pdifftabooaugment{U}{x}{y}{z}}{\sigma}{Y}\) (no $z$ occurs),
and
\(w\in\usubstapp{\pdifftabooaugment{U}{x}{y}{z}}{\sigma}{Z}\) the result of uniformly renaming $z$ to $w$ in \(z\in\usubstapp{\pdifftabooaugment{U}{x}{y}{z}}{\sigma}{Z}\), where $y$ does not occur.
Without loss of generality (by performing two subsequent uniform substitutions), no symbol that is being replaced by $\sigma$ occurs in any of $\sigma$'s replacements.
Hence, $\sigma$ is idempotent and
\(\iwinreg[\pdiffgame{\D{x}=\urename[\usubstapp{\pdifftabooaugment{U}{x}{y}{z}}{\sigma}{\genDG{x}{y}{z}}]{y}{v}\urename{z}{w}}{v\in\usubstapp{\pdifftabooaugment{U}{x}{y}{z}}{\sigma}{Y}}{w\in\usubstapp{\pdifftabooaugment{U}{x}{y}{z}}{\sigma}{Z}}]{\Ie}{X}\)
= \(\iwinreg[\pdiffgame{\D{x}=\urename[\usubstapp{\pdifftabooaugment{U}{x}{y}{z}}{\sigma}{\genDG{x}{y}{z}}]{y}{v}\urename{z}{w}}{v\in\usubstapp{\pdifftabooaugment{U}{x}{y}{z}}{\sigma}{Y}}{w\in\usubstapp{\pdifftabooaugment{U}{x}{y}{z}}{\sigma}{Z}}]{\Iae}{X}\).
Now that both sides are phrased in the same interpretation, the desired equivalence that
\(\iwin[\pdiffgame{\D{x}=\genDG{x}{y}{z}}{y\in Y}{z\in Z}]{\Iae}{X}\)
iff
\(\iwin[\pdiffgame{\D{x}=\urename[\usubstapp{\pdifftabooaugment{U}{x}{y}{z}}{\sigma}{\genDG{x}{y}{z}}]{y}{v}\urename{z}{w}}{v\in\usubstapp{\pdifftabooaugment{U}{x}{y}{z}}{\sigma}{Y}}{w\in\usubstapp{\pdifftabooaugment{U}{x}{y}{z}}{\sigma}{Z}}]{\Iae}{X}\)
follows provided that the following \dGL formula is true in $\iportray{\Iae}$ for a fresh game symbol $c$ with \(\imodel{\Iae}{\ddiamond{c}{\ltrue}}=X\):
\begin{equation}
\ddiamond{\pdiffgame{\D{x}=\urename[\usubstapp{\pdifftabooaugment{U}{x}{y}{z}}{\sigma}{\genDG{x}{y}{z}}]{y}{v}\urename{z}{w}}{v\in\usubstapp{\pdifftabooaugment{U}{x}{y}{z}}{\sigma}{Y}}{w\in\usubstapp{\pdifftabooaugment{U}{x}{y}{z}}{\sigma}{Z}}}{\ddiamond{c}{\ltrue}}
\lbisubjunct
\ddiamond{\pdiffgame{\D{x}=\genDG{x}{y}{z}}{y\in Y}{z\in Z}}{\ddiamond{c}{\ltrue}}
\label{eq:diffgame-usubst-key-equiv}
\end{equation}
Without loss of generality, replace free occurrences of variables \(\scomplement{\{x,\D{x},y,\D{y},z,\D{z}\}}\) by their respective real values in $\iget[state]{\Iae}$.
Now \rref{eq:diffgame-usubst-key-equiv} is true in $\iportray{\Iae}$ by the (locally sound) differential game refinement proof schema \cite{DBLP:journals/tocl/Platzer17} for $\ddiamond{}{}$ once per implication:
\[
\cinferenceRule[diffgamerefined|DGR]{differential game refinement, diamond dual}
  {\linferenceRule[sequent]
    {\lforall{y\in Y}{\lexists{v\in V}{\lforall{w\in W}{\lexists{z\in Z}{\lforall{x}{(\genDGb{x}{v}{w}=\genDG{x}{y}{z})}}}}}}
    {\ddiamond{\pdiffgame{\D{x}=\genDGb{x}{v}{w}}{v\in V}{w\in W}}{F} \limply \ddiamond{\pdiffgame{\D{x}=\genDG{x}{y}{z}}{y\in Y}{z\in Z}}{F}}
  }
  {}%
\]
By rule \irref{diffgamerefined} for both implications of \rref{eq:diffgame-usubst-key-equiv}, it suffices to show validity in $\iget[const]{\Iae}$ of:
\begin{equation}
\begin{aligned}
  &\lforall{y \in Y}{\lexists{v\in\usubstapp{\pdifftabooaugment{U}{x}{y}{z}}{\sigma}{Y}}{\lforall{w\in\usubstapp{\pdifftabooaugment{U}{x}{y}{z}}{\sigma}{Z}}{\lexists{z\in Z}{\lforall{x}{(
  \urename[\usubstapp{\pdifftabooaugment{U}{x}{y}{z}}{\sigma}{\genDG{x}{y}{z}}]{y}{v}\urename{z}{w}
  = \genDG{x}{y}{z}
  )}}}}}\\
  &\lforall{v\in\usubstapp{\pdifftabooaugment{U}{x}{y}{z}}{\sigma}{Y}}{\lexists{y\in Y}{\lforall{z\in Z}{\lexists{w\in\usubstapp{\pdifftabooaugment{U}{x}{y}{z}}{\sigma}{Z}}{\lforall{x}{(
  \urename[\usubstapp{\pdifftabooaugment{U}{x}{y}{z}}{\sigma}{\genDG{x}{y}{z}}]{y}{v}\urename{z}{w}
  = \genDG{x}{y}{z} 
  )}}}}}
\end{aligned}
\label{eq:diffgame-usubst-key}
\end{equation}
Both formulas are shown with $v=y$ and $w=z$ as witnesses.
By \rref{lem:usubst} all $\pdifftabooaugment{U}{x}{y}{z}$-variations $\iget[state]{\Iaz}$ of $\iget[state]{\I}$ satisfy
\(\imodels{\Iaz}{Y}\)
iff \(\imodels{\Iz}{\usubstapp{\pdifftabooaugment{U}{x}{y}{z}}{\sigma}{Y}}\)
iff, as $\sigma$ idempotent, \(\imodels{\Iaz}{\usubstapp{\pdifftabooaugment{U}{x}{y}{z}}{\sigma}{Y}}\)
iff, by uniform renaming and \rref{lem:coincidence} as $\D{y}$ is not in $\usubstapp{\pdifftabooaugment{U}{x}{y}{z}}{\sigma}{Y}$, \(\imodels{\Iazyv}{\urename[\urengroup{\usubstapp{\pdifftabooaugment{U}{x}{y}{z}}{\sigma}{Y}}]{y}{v}}\)
= \(\imodel{\Iazyv}{v\in\usubstapp{\pdifftabooaugment{U}{x}{y}{z}}{\sigma}{Y}}\).
Here, $\iget[state]{\Iazyv}$ is the state ${\modif{\iget[state]{\Iz}}{v}{\iget[state]{\Iz}(y)}}$ as in \rref{eq:urenameadjoint} of \rref{lem:brename}, where $y,\D{y},\D{v}$ do not occur in \(\urename[\urengroup{\usubstapp{\pdifftabooaugment{U}{x}{y}{z}}{\sigma}{Y}}]{y}{v}\).
By a similar argument: 
\(\imodels{\Iaz}{Z}\)
iff \(\imodels{\Iazzw}{\urename[\urengroup{\usubstapp{\pdifftabooaugment{U}{x}{y}{z}}{\sigma}{Z}}]{z}{w}}\)
= \(\imodel{\Iazzw}{w\in\usubstapp{\pdifftabooaugment{U}{x}{y}{z}}{\sigma}{Z}}\).
When \(v=y\) and \(w=z\), the constraints of \rref{eq:diffgame-usubst-key} are met in a state of $\iget[const]{\Iae}$ for $y,z$ iff they are met for $v,w$.

Finally, by \rref{lem:usubst-term} when $\iget[state]{\Iaz}$ is a $\pdifftabooaugment{U}{x}{y}{z}$-variation of $\iget[state]{\I}$:
\(\ivaluation{\Iaz}{\genDG{x}{y}{z}}\)
= \(\ivaluation{\Iz}{\usubstapp{\pdifftabooaugment{U}{x}{y}{z}}{\sigma}{\genDG{x}{y}{z}}}\)
which by uniform renaming and \rref{lem:coincidence-term} as $\D{y}$ and $\D{z}$ are not in \(\usubstapp{\pdifftabooaugment{U}{x}{y}{z}}{\sigma}{\genDG{x}{y}{z}}\) equals
\(\ivaluation{\Izyvzw}{\urename[\urengroup{\usubstapp{\pdifftabooaugment{U}{x}{y}{z}}{\sigma}{\genDG{x}{y}{z}}}]{y}{v}\urename{z}{w}}\),
which by idempotence of $\sigma$ equals
\(\ivaluation{\Iazyvzw}{\urename[\urengroup{\usubstapp{\pdifftabooaugment{U}{x}{y}{z}}{\sigma}{\genDG{x}{y}{z}}}]{y}{v}\urename{z}{w}}\).
Thus, the states $\iget[state]{\Iaz}$ that are $\{x,y,z,v,w\}$-variations of $\iget[state]{\Ie}$ so $\pdifftabooaugment{U}{x}{y}{z}$-variation of $\iget[state]{\I}$ satisfying \(\imodels{\Iaz}{v=y \land w=z}\) witness \rref{eq:diffgame-usubst-key}, because
\(\imodels{\Iaz}{\urename[\urengroup{\usubstapp{\pdifftabooaugment{U}{x}{y}{z}}{\sigma}{\genDG{x}{y}{z}}}]{y}{v}\urename{z}{w}=\genDG{x}{y}{z}}\) 
by \rref{lem:coincidence-term} as $v,w$ are not in \(\genDG{x}{y}{z}\),
for all values of $y,w,x$ (with $v:=y,z:=w$),
or for all values of $v,z,x$ ($y:=v,w:=z$), respectively.
\end{proofatend}

The proof makes clever use of differential game refinements \cite{DBLP:journals/tocl/Platzer17} to avoid the significant complexities and semantic subtleties of differential games.

\section{Conclusion}
This paper introduced significantly faster uniform substitution mechanisms, the dominant logical inference in axiomatic small core hybrid systems/games provers.
It is also first in proving soundness of uniform substitution for differential games.

Implementations exhibit a linear runtime complexity compared to the exponential complexity that direct implementations \cite{DBLP:conf/cade/FultonMQVP15} of prior Church-style uniform substitutions exhibit, except when applying aggressive space/time optimization tradeoffs where that drops down to a quadratic runtime in practice.

\section*{Acknowledgment}
I thank Frank Pfenning for useful discussions and the anonymous reviewers for their helpful feedback.
I appreciate the kind advice of the Isabelle group at TU Munich for the subsequent formalization \cite{DBLP:journals/afp/Platzer19} of the proofs.

This research is supported by the Alexander von Humboldt Foundation and by the AFOSR under grant number FA9550-16-1-0288.

\renewcommand{\doi}[1]{doi:\href{https://doi.org/#1}{\nolinkurl{#1}}}

The views and conclusions contained in this document are those of the author and should not be interpreted as representing the official policies, either expressed or implied, of any sponsoring institution, the U.S. government or any other entity.

\bibliographystyle{plainurl}
\bibliography{platzer,bibliography}

\begin{thebibliography}{10}

\bibitem{KeYBook2016}
Wolfgang Ahrendt, Bernhard Beckert, Richard Bubel, Reiner H{\"a}hnle, Peter~H.
  Schmitt, and Matthias Ulbrich, editors.
\newblock {\em Deductive Software Verification -- The {KeY} Book}, volume 10001
  of {\em LNCS}.
\newblock Springer, 2016.
\newblock \href {http://dx.doi.org/10.1007/978-3-319-49812-6}
  {\path{doi:10.1007/978-3-319-49812-6}}.

\bibitem{DBLP:conf/cpp/BohrerRVVP17}
Brandon Bohrer, Vincent Rahli, Ivana Vukotic, Marcus V{\"o}lp, and Andr{\'{e}}
  Platzer.
\newblock Formally verified differential dynamic logic.
\newblock In Yves Bertot and Viktor Vafeiadis, editors, {\em Certified Programs
  and Proofs - 6th ACM SIGPLAN Conference, CPP 2017, Paris, France, January
  16-17, 2017}, pages 208--221, New York, 2017. ACM.
\newblock \href {http://dx.doi.org/10.1145/3018610.3018616}
  {\path{doi:10.1145/3018610.3018616}}.

\bibitem{DBLP:journals/jsl/Church40}
Alonzo Church.
\newblock A formulation of the simple theory of types.
\newblock {\em J. Symb. Log.}, 5(2):56--68, 1940.
\newblock \href {http://dx.doi.org/10.2307/2266170}
  {\path{doi:10.2307/2266170}}.

\bibitem{Church_1956}
Alonzo Church.
\newblock {\em Introduction to Mathematical Logic}.
\newblock Princeton University Press, Princeton, 1956.

\bibitem{DeBruijn72}
N.G de~Bruijn.
\newblock Lambda calculus notation with nameless dummies, a tool for automatic
  formula manipulation, with application to the {Church-Rosser} theorem.
\newblock {\em Indagationes Mathematicae}, 75(5):381 -- 392, 1972.
\newblock \href {http://dx.doi.org/10.1016/1385-7258(72)90034-0}
  {\path{doi:10.1016/1385-7258(72)90034-0}}.

\bibitem{DBLP:journals/tams/ElliottK74}
Robert~J. Elliott and Nigel~J. Kalton.
\newblock {Cauchy} problems for certain {Isaacs-Bellman} equations and games of
  survival.
\newblock {\em Trans. Amer. Math. Soc.}, 198:45--72, 1974.
\newblock \href {http://dx.doi.org/10.1090/S0002-9947-1974-0347383-8}
  {\path{doi:10.1090/S0002-9947-1974-0347383-8}}.

\bibitem{DBLP:journals/indianam/EvansSouganidis84}
Lawrence~Craig Evans and Panagiotis~E. Souganidis.
\newblock Differential games and representation formulas for solutions of
  {Hamilton-Jacobi-Isaacs} equations.
\newblock {\em Indiana Univ. Math. J.}, 33(5):773--797, 1984.
\newblock \href {http://dx.doi.org/10.1512/iumj.1984.33.33040}
  {\path{doi:10.1512/iumj.1984.33.33040}}.

\bibitem{DBLP:conf/cade/FultonMQVP15}
Nathan Fulton, Stefan Mitsch, Jan-David Quesel, Marcus V{\"o}lp, and
  Andr{\'{e}} Platzer.
\newblock {KeYmaera X}: An axiomatic tactical theorem prover for hybrid
  systems.
\newblock In Amy Felty and Aart Middeldorp, editors, {\em CADE}, volume 9195 of
  {\em LNCS}, pages 527--538, Berlin, 2015. Springer.
\newblock \href {http://dx.doi.org/10.1007/978-3-319-21401-6_36}
  {\path{doi:10.1007/978-3-319-21401-6_36}}.

\bibitem{Harel_et_al_2000}
David Harel, Dexter Kozen, and Jerzy Tiuryn.
\newblock {\em Dynamic Logic}.
\newblock MIT Press, Cambridge, 2000.
\newblock \href {http://dx.doi.org/10.7551/mitpress/2516.001.0001}
  {\path{doi:10.7551/mitpress/2516.001.0001}}.

\bibitem{DBLP:journals/jsl/Henkin53}
Leon Henkin.
\newblock Banishing the rule of substitution for functional variables.
\newblock {\em J. Symb. Log.}, 18(3):pp. 201--208, 1953.
\newblock \href {http://dx.doi.org/10.2307/2267403}
  {\path{doi:10.2307/2267403}}.

\bibitem{HilbertAckermann28}
David Hilbert and Wilhelm Ackermann.
\newblock {\em Grundz{\"u}ge der theoretischen Logik}.
\newblock Springer, Berlin, 1928.

\bibitem{HilbertBernays34}
David Hilbert and Paul Bernays.
\newblock {\em Grundlagen der Mathematik}, volume~I.
\newblock Springer, 2 edition, 1934.

\bibitem{DBLP:journals/tac/MitchellBT05}
Ian Mitchell, Alexandre~M. Bayen, and Claire Tomlin.
\newblock A time-dependent {Hamilton-Jacobi} formulation of reachable sets for
  continuous dynamic games.
\newblock {\em IEEE T. Automat. Contr.}, 50(7):947--957, 2005.
\newblock \href {http://dx.doi.org/10.1109/TAC.2005.851439}
  {\path{doi:10.1109/TAC.2005.851439}}.

\bibitem{DBLP:conf/pldi/PfenningE88}
Frank Pfenning and Conal Elliott.
\newblock Higher-order abstract syntax.
\newblock In Richard~L. Wexelblat, editor, {\em PLDI}, pages 199--208. ACM,
  1988.
\newblock \href {http://dx.doi.org/10.1145/53990.54010}
  {\path{doi:10.1145/53990.54010}}.

\bibitem{DBLP:journals/tocl/Platzer15}
Andr{\'{e}} Platzer.
\newblock Differential game logic.
\newblock {\em {ACM} Trans. Comput. Log.}, 17(1):1:1--1:51, 2015.
\newblock \href {http://dx.doi.org/10.1145/2817824}
  {\path{doi:10.1145/2817824}}.

\bibitem{DBLP:journals/jar/Platzer17}
Andr{\'{e}} Platzer.
\newblock A complete uniform substitution calculus for differential dynamic
  logic.
\newblock {\em J. Autom. Reas.}, 59(2):219--265, 2017.
\newblock \href {http://dx.doi.org/10.1007/s10817-016-9385-1}
  {\path{doi:10.1007/s10817-016-9385-1}}.

\bibitem{DBLP:journals/tocl/Platzer17}
Andr{\'{e}} Platzer.
\newblock Differential hybrid games.
\newblock {\em {ACM} Trans. Comput. Log.}, 18(3):19:1--19:44, 2017.
\newblock \href {http://dx.doi.org/10.1145/3091123}
  {\path{doi:10.1145/3091123}}.

\bibitem{DBLP:conf/cade/Platzer18}
Andr{\'{e}} Platzer.
\newblock Uniform substitution for differential game logic.
\newblock In Didier Galmiche, Stephan Schulz, and Roberto Sebastiani, editors,
  {\em IJCAR}, volume 10900 of {\em LNCS}, pages 211--227. Springer, 2018.
\newblock \href {http://dx.doi.org/10.1007/978-3-319-94205-6_15}
  {\path{doi:10.1007/978-3-319-94205-6_15}}.

\bibitem{DBLP:journals/afp/Platzer19}
Andr{\'{e}} Platzer.
\newblock Differential game logic.
\newblock {\em Archive of Formal Proofs}, 2019, 2019.
\newblock Formal proof development.
\newblock URL: \url{http://isa-afp.org/entries/Differential\_Game\_Logic.html}.

\bibitem{DBLP:conf/cade/QueselP12}
Jan-David Quesel and Andr{\'{e}} Platzer.
\newblock Playing hybrid games with {KeYmaera}.
\newblock In Bernhard Gramlich, Dale Miller, and Ulrike Sattler, editors, {\em
  IJCAR}, volume 7364 of {\em LNCS}, pages 439--453, Berlin, 2012. Springer.
\newblock \href {http://dx.doi.org/10.1007/978-3-642-31365-3_34}
  {\path{doi:10.1007/978-3-642-31365-3_34}}.

\bibitem{Quine34}
Willard Van~Orman Quine.
\newblock {\em A System of Logistic}.
\newblock Harvard Univ. Press, 1934.

\bibitem{DBLP:journals/ndjfl/Schneider80}
Hubert~H. Schneider.
\newblock Substitutions for predicate variables and functional variables.
\newblock {\em Notre Dame J. Formal Logic}, 21(1):33--44, 01 1980.
\newblock \href {http://dx.doi.org/10.1305/ndjfl/1093882937}
  {\path{doi:10.1305/ndjfl/1093882937}}.

\end{thebibliography}

\end{document}